\documentclass{article}

\usepackage[affil-it]{authblk}
\usepackage[usenames,dvipsnames]{xcolor}
\usepackage{amsfonts}
\usepackage{amsmath,amsthm,amssymb,dsfont}
\usepackage{enumerate}
\usepackage[english]{babel}
\usepackage{graphicx}	
\usepackage{subcaption}
\usepackage[margin=3cm]{geometry}
\usepackage{url}
\usepackage{todonotes}
\usepackage{bbm}
\usepackage{pifont}

\usepackage{tikz}
\usetikzlibrary{chains}
\usetikzlibrary{fit}

\usepackage{makecell}

\usepackage{epsfig}
\usetikzlibrary{shapes.symbols,patterns} 
\usepackage{pgfplots}

\usepackage{hyperref}
\hypersetup{colorlinks=true,citecolor=blue,linkcolor=blue,filecolor=blue,urlcolor=blue,breaklinks=true}

\usepackage{nicefrac}
\usepackage{mathtools}

\usepackage{algorithm}
\usepackage{algorithmic}

\theoremstyle{plain}
\newtheorem{theorem}{Theorem}[section]
\newtheorem{lemma}[theorem]{Lemma}

\newtheorem{corollary}[theorem]{Corollary}
\newtheorem{proposition}[theorem]{Proposition}

\theoremstyle{definition}

\newtheorem{remark}[theorem]{Remark}
\newtheorem{example}[theorem]{Example}
\newtheorem{assumption}[theorem]{Assumption}

\newcommand*{\cA}{\mathcal{A}}

\newcommand*{\cC}{\mathcal{C}}

\newcommand*{\cS}{\mathcal{S}}

\newcommand*{\cX}{\mathcal{X}}

\newcommand*{\N}{\mathbb{N}}

\newcommand*{\R}{\mathbb{R}}

\newcommand*{\C}{\mathbb{C}}

\newcommand*{\eps}{\varepsilon}
\newcommand*{\diag}{\mathrm{diag}}

\newcommand*{\poly}{\mathrm{poly}}
\newcommand*{\polylog}{\mathrm{polylog}}

\newcommand*{\tr}{\mathrm{tr}}
\newcommand*{\ket}[1]{| #1 \rangle}
\newcommand*{\bra}[1]{\langle #1 |}

\newcommand{\norm}[1]{\left\lVert#1\right\rVert}

\newcommand{\splitatcommas}[1]{%
  \begingroup
  \begingroup\lccode`~=`, \lowercase{\endgroup
    \edef~{\mathchar\the\mathcode`, \penalty0 \noexpand\hspace{0pt plus 1em}}%
  }\mathcode`,="8000 #1%
  \endgroup
}

\newcommand{\cmark}{\ding{51}}%
\newcommand{\xmark}{\ding{55}}%

\DeclareSymbolFont{mymathoperators}{OT1}{phv}{m}{n}

\DeclareMathSymbol{\protoast}{\mathbin}{mymathoperators}{"2A}

\renewcommand*{\ast}{\mathbin{\raisebox{-0.7ex}{\ensuremath{\protoast}}}}

\title{Quantum Legendre-Fenchel Transform}

\author{\normalsize David Sutter$^{1}$, Giacomo Nannicini$^{2}$, Tobias Sutter$^{3}$, and Stefan Woerner$^{1}$}
\affil{\small $^{1}$IBM Quantum, IBM Research -- Zurich\\
\small $^{2}$IBM Quantum, IBM T.J.~Watson Research Center\\
\small $^{3}$Risk Analytics and Optimization Chair, EPFL
}
\date{}

\begin{document}

\maketitle

\begin{abstract}
We present a quantum algorithm to compute the discrete Legendre-Fenchel transform. Given access to a convex function evaluated at $N$ points, the algorithm outputs a quantum-mechanical representation of its corresponding discrete Legendre-Fenchel transform evaluated at $K$ points in the transformed space. For a fixed regular discretization of the dual space the expected running time scales as $O(\sqrt{\kappa}\,\polylog(N,K))$, where $\kappa$ is the condition number of the function. If the discretization of the dual space is chosen adaptively with $K$ equal to $N$, the running time reduces to $O(\polylog(N))$. We explain how to extend the presented algorithm to the multivariate setting and prove lower bounds for the query complexity, showing that our quantum algorithm is optimal up to polylogarithmic factors. For multivariate functions with $\kappa=1$, the quantum algorithm computes a quantum-mechanical representation of the Legendre-Fenchel transform at $K$ points exponentially faster than any classical algorithm can compute it at a single point.
\end{abstract}


\paragraph{Update:}
We recently discovered an error in the correctness proof of Algorithm~\ref{algo_QLFT_2d}, and we have not found a way to fix it yet. We believe that our hardness results and lower bounds are correct, but the running time of our quantum algorithm does not appear to be so, therefore we are unsure if there is a quantum speedup in this setting. We will update this manuscript once we have fully resolved this issue.

\section{Introduction}
Quantum algorithms utilize intrinsic features of quantum mechanics to perform certain tasks asymptotically faster than any known classical method. The design of quantum algorithms is challenging; as a result, most currently known quantum algorithms use a few basic building blocks as subroutines. Examples of these building blocks are: the quantum Fourier transform~\cite{QFT94,nielsenChuang_book}, Grover's algorithm for unstructured search~\cite{grover96,Grover97}, or quantum phase estimation~\cite{Kitaev_qpe}. It is a major scientific challenge to identify more quantum algorithms that outperform their classical counterparts. In particular, the discovery of novel quantum subroutines could help the development of further quantum algorithms.

In this paper we present a quantum algorithm for a task that is known in the classical world, but that has not been quantized before: the computation of the \emph{Legendre-Fenchel transform} (LFT). The LFT, also known as \emph{convex conjugate} or simply as \emph{Legendre transform}, is used in many different disciplines. An example of such a discipline is thermodynamics, where the LFT is used to switch between potentials~\cite{reichl_book}.
More precisely, an important task in thermodynamics is the computation of potentials, e.g., the internal energy. Potentials are defined as expectation values of the form $\langle U \rangle (v) = \tr[\rho(v) U]$, where $\rho(v)$ denotes a state that depends on a variable $v$, e.g., the volume of a gas, and $U$ is a Hermitian operator for a quantity such as the internal energy. The LFT of a potential defines another thermodynamic potential, e.g., the enthalpy, where the variable describing the volume is transformed into a (dual) variable describing the pressure of the gas. Depending on the problem task, it can be very helpful to consider different potentials, hence the usefulness of the LFT.

The LFT finds many more applications in science, such as in classical mechanics, where the LFT serves as a bridge between the Hamiltonian and the Lagrangian formalism~\cite{goldstein_book} and large deviation theory, where the LFT provides the link between the log-moment generating function, and a corresponding rate function via Cram\'er's theorem~\cite[Section~2.2]{LDT_book}.  
In convex analysis, the LFT has an importance similar to that of the Fourier transform in signal processing~\cite{rockafellar70}. 
One of the main reasons behind the fundamental relevance of the LFT in convex analysis is the fact that the \emph{infimal convolution} (also called \emph{epi-sum}) is equivalent to a simple addition in the Legendre-Fenchel space~\cite[Theorem~16.4]{rockafellar70}, whereas the analogous property for the Fourier transform is that the convolution operator is equivalent to a product in the Fourier space.
In mathematical optimization the LFT is used to establish the duality theory which serves as a certificate to ensure that a certain optimizer achieves the optimal value. Figure~\ref{fig_overview} gives an overview about some applications of the LFT.
\begin{figure}[!htb]
    \centering
    \begin{subfigure}[t]{0.3\textwidth}
    \hspace{-5mm}
    \scalebox{0.9}{
        \begin{tikzpicture}
        \node at (0,0) {Internal energy};
        \node (U) at (0,-0.75) {$\langle U\rangle (v)$};
        \node at (3.5,0) {Enthalpy};
        \node (H) at (3.5,-0.75) {$\langle H\rangle (p)$};
        \draw[thick,->] (U.north east) to [out=30,in=150] (H.north west);     
        \draw[thick,<-] (U.south east) to [out=-30,in=-150] (H.south west);     
        \node at (1.75,-0.55) {\small{$\mathrm{LFT}_{v\to p}$}};
        \node at (1.75,-0.95) {\small{$\mathrm{LFT}_{p\to v}$}};           
        \end{tikzpicture}
        }
        \caption{Thermodynamics}
        \label{fig:gull}
    \end{subfigure} 
    \hspace{3mm}
        \begin{subfigure}[t]{0.3\textwidth}
    \scalebox{0.9}{
        \begin{tikzpicture}
        \node at (0,0) {Lagrangian};
        \node (L) at (0,-0.75) {$L(v)$};
        \node at (3,0) {Hamiltonian};
        \node (H) at (3,-0.75) {$H(p)$};
        \draw[thick,->] (L.north east) to [out=30,in=150] (H.north west);
        \draw[thick,<-] (L.south east) to [out=-30,in=-150] (H.south west);
        \node at (1.5,-0.55) {\small{$\mathrm{LFT}_{v\to p}$}};
        \node at (1.5,-1) {\small{$\mathrm{LFT}_{p\to v}$}};
        \end{tikzpicture}
        }
        \caption{Mechanics}
        \label{fig:gull2}
    \end{subfigure}    
        \hspace{1.5mm}   
        \begin{subfigure}[t]{0.3\textwidth}  
        \vspace{-15.3mm}
    \scalebox{0.9}{
        \begin{tikzpicture}
        \def \y{0.2}
        \node at (0,0+\y) {Primal problem};
        \node at (0,-0.75) {$\inf\limits_{x\in \R^d} f(x)$};
        \node at (1.3,-0.7) {$\geq$};
        \node at (3,0+\y) {Dual problem};
        \node at (3,-0.75) {$\sup\limits_{s \in \R^d}\{-f^{\ast}(s) \}$};
        \end{tikzpicture}
        }
        \caption{Optimization}
    \end{subfigure}  
 \caption{Schematic overview of different applications of the LFT in three areas: (a) In thermodynamics the LFT is utilized to switch between potentials~\cite{reichl_book}; e.g.~considering the LFT of the internal energy with respect to the volume gives the enthalpy, where $p$ is the pressure.
 (b) In classical mechanics the LFT serves as a bridge between the Hamiltonian and the Lagrangian formalism~\cite{goldstein_book}; the LFT of a typical Lagrangian $L(v)$ gives the Hamiltonian $H(p)$, where $v$ and $p$ are local coordinates. 
 (c) In mathematical optimization the LFT is a building block in the duality theory that is used to verify optimality; the negative LFT of an arbitrary convex function $f$ that we want to minimize serves as a generic lower bound.
 }            \label{fig_overview}  
\end{figure}
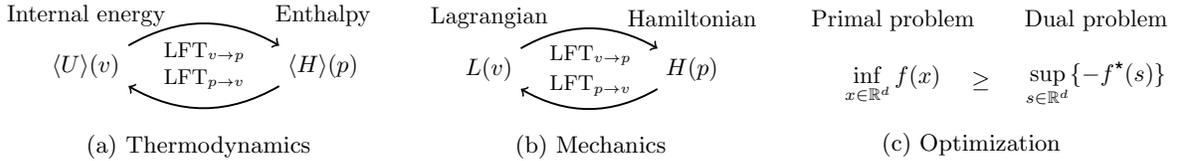   

The LFT can also be utilized to efficiently solve a large class of optimization problems called \emph{dynamic programs}~\cite{peyman,QDP20}. Dynamic programs are arguably the most commonly used tool for sequential decision-making. The quantum LFT introduced in this paper can be used as a crucial subroutine to design a quantum algorithm that can solve some dynamic programs satisfying certain convexity assumptions with a quadratic speedup compared to the classical Bellman approach~\cite{QDP20,bellman2003dynamic}.

For a function $f: \R^d \to \R$, its LFT is denoted $f^{\ast}: \R^d \to \R$ and defined as
\begin{align} \label{eq_contLFT}
f^{\ast}(s) := \sup_{x \in \R^d}\{\langle s, x \rangle - f(x) \}\, .
\end{align}
The LFT is a nonlinear operation that has many desirable properties, some of which are summarized in Section~\ref{sec_propLFT}. For certain functions the LFT can be computed in closed-form; however, in general the supremum in~\eqref{eq_contLFT} may not allow a simple closed-form solution, hence it is not straightforward to evaluate. Furthermore, in some applications the function $f$ may not be known analytically, but only observed via a finite number of data points. 

This motivates the definition of a \emph{discrete Legendre-Fenchel transform}. 
Let $\cX^d_N=\{x_0,\ldots, x_{N-1}\} \subseteq \R^d$ and $\cS^d_K=\{s_0,\ldots, s_{K-1}\}\subseteq \R^d$ be discrete primal and dual spaces, respectively. Then, the discrete LFT of the function $f$ is defined by the mapping $(f(x_0),\ldots,f(x_{N-1})) \mapsto (f^*(s_0),\ldots,f^*(s_{K-1}))$,\footnote{We use the two different asterisk symbols \raisebox{-0.425mm}{\large{$\ast$}} and $*$ to distinguish between the continuous and the discerte LFT.} where
\begin{align} \label{eq_DLFT}
f^*(s_j) := \max_{x \in \cX^d_N} \{\langle s_j, x\rangle -f(x) \}\, , \qquad \text{for} \quad j=0,\ldots,K-1 \, .
\end{align}
If $f$ is continuous, the discrete LFT converges to the (continuous) LFT when $N,K \to \infty$, if we assume that $f$ and $f^*$ have a bounded and fixed domain~\cite[Theorem~2.1]{corrias96} (see also Lemma~\ref{lem_discLFT_contLFT}). The discrete LFT plays a fundamental role in discrete convex analysis \cite{murota2003book}, where it defines a conjugacy relationship between some well-studied classes of functions ($L$-convex and $M$-convex, $L^\natural$-convex and $M^\natural$-convex) obtained by discretizing different characterizations of continuous convex functions. For example, the discrete LFT can be used to define primal/dual optimality criteria for submodular flow problems \cite{iwata2002conjugate,murota1999submodular}.

Even though the LFT is defined for an arbitrary function $f$, in this work we impose some mild regularity assumptions that will be crucial for the quantum algorithm presented in Section~\ref{sec_QLFT_1d}. The precise assumptions are discussed in Section~\ref{sec_regularity}; arguably the strongest one is that we restrict ourselves to convex functions. Furthermore, we restrict ourselves to functions with a compact domain, assumed to be $[0,1]^d$ without loss of generality.

From~\eqref{eq_DLFT} we immediately see that a brute force calculation of the discrete LFT has complexity $O(N K)$. A more refined algorithm that requires $O(N+K)$ operations only is given in~\cite{lucet97}. We explain this algorithm in detail in Section~\ref{sec_DLFT}.
Since writing down the input $(f(x_0),\ldots,f(x_{N-1}))$ and output $(f^*(s_0),\ldots,f^*(s_{K-1}))$ takes $O(N)$ and $O(K)$ time, respectively, the classical algorithm is asymptotically optimal.  
In multidimensional scenarios (i.e., $d>1$) the number of discretization points $N$ typically needs to be exponentially large in $d$, i.e., $N \propto \exp(d)$, to ensure a good approximation of the function's behavior over $[0,1]^d$. This is one of the sources of the well-known ``curse of dimensionality", leading to a running time for the LFT algorithm that is exponential in $d$.

\paragraph{Results}
For any convex function $f:[0,1]\to \R$ satisfying some mild regularity assumptions (see Section~\ref{sec_regularity}) we present two algorithms to compute a quantum-mechanical representation of the discrete LFT. Given access to $f$ on a regular discretization in the form of samples $(f(x_0),\ldots,f(x_{N-1}))$, these algorithms construct the quantum state
\begin{align*}
    \frac{1}{\sqrt{K}} \sum_{j=0}^{K-1} \ket{j} \ket{f^*(s_j)} \ket{s_{j}} \, .
\end{align*}
The two algorithms differ in the sense that:
\begin{enumerate}[(a)]
    \item one uses a regular discretization of the dual space and its expected running time is given by $O(\sqrt{\kappa}\, \polylog(N,K))$, where $\kappa$ is the condition number of $f$;\footnote{In the running time expression there are no products between $N$ and $K$.}
    \item the other uses an adaptive discretization of the dual space of size $N$, chosen by the algorithm and hence a-priori unknown, and runs in time $O(\polylog(N))$.\footnote{An adaptive discretization means that the algorithm chooses the dual space based on the input, in such a way that each dual point $s_j$ has a unique distinct optimizer in~\eqref{eq_DLFT}.}
\end{enumerate}
We refer to Theorems~\ref{thm_QLFT} and~\ref{thm_QLFT_adaptive} for more precise statements.

Our algorithms can be extended to the multidimensional case. For $n \in \N$ we use the notation $[n]:=\{0,1,\ldots,n-1\}$. 
For any convex function $f:[0,1]^d\to \R$ satisfying some mild regularity assumptions (see Section~\ref{sec_regularity}), known via $N=\prod_{\ell=0}^{d-1} N_\ell$ discrete samples $\{f(x_{i_0},\ldots x_{i_{d-1}})\}_{i_{\ell}\in [N_\ell]}$ for $\ell \in[d]$ on a regular grid, we can compute the state
\begin{align*}
    \frac{1}{\sqrt{K}} \sum_{j_0=0}^{K_0-1}\ldots \sum_{j_{d-1}=0}^{K_{d-1}-1}\ket{j_0,\ldots,j_{d-1}} \ket{f^*(s_{j_0},\ldots,s_{j_{d-1}})} \ket{s_{j_0},\ldots,s_{j_{d-1}}} \, ,
\end{align*}
 where $K=\prod_{\ell=0}^{d-1} K_\ell$ denotes the discretization of the dual space.
As for the one-dimensional case, the two algorithms are different in the following sense:
\begin{enumerate}[(a)]
    \item one uses a regular discretization of the dual space and its expected running time is given by $O(\kappa^{d/2}\,\polylog(N,K))$;
    \item the other uses an adaptive discretization of the dual space with size $N$, chosen by the algorithm, and runs in $O(\polylog(N))$.
\end{enumerate}
We refer to Corollaries~\ref{cor_d_dim_QLFT_regular} and~\ref{cor_d_QLFT_adaptive} for more precise statements in the multidimensional case. Figure~\ref{table_overview} gives an overview about the different algorithms for the QLFT.
\begin{table}[htb!]
\centering
\begin{tabular}{|c|c c c c c c|} 
 \hline
  & dim. & primal space & dual space & expected complexity & process & speedup   \\ 
 \hline
 Algorithm~\ref{algo_QLFT} & 1 & $N$-regular &  $K$-regular & $\sqrt{\kappa}\,\polylog(N,K)$ & stochastic & \xmark  \\ 
 Algorithm~\ref{algo_QLFT_adaptive} & 1 & $N$-regular & $N$-adaptive & $\polylog(N)$ & unitary &  \xmark \\ 
 Algorithm~\ref{algo_QLFT_2d} & $d$ & $N$-regular & $K$-regular & $\kappa^{d/2}\,\polylog(N,K)$ & stochastic & \cmark \\ 
 Algorithm~\ref{algo_QLFT_2d_adaptive} & $d$ & $N$-regular & $N$-adaptive & $\polylog(N)$ & unitary & \cmark \\
 \hline
\end{tabular}
\caption{Overview of different algorithms for the QLFT. In the multidimensional setting a quantum speedup is possible whenever the condition number $\kappa\geq 1$ is sufficiently small. An exponential speedup is achieved for functions with $\kappa=1$.}
\label{table_overview}
\end{table}

\paragraph{Quantum speedup and optimality}
We analyze the classical computational complexity of computing $f^*(s)$ at some given $s\in \cS^d_K$ for $K=2^d$ and show that any classical algorithm requires time $\Omega(2^d/d)$ (see Proposition~\ref{prop_dltlowerbound}). The same lower bound holds for the task of sampling $f^*(s)$ uniformly at random from a set $\cS^d_K$ when $K=2^d$ (see Proposition~\ref{prop_sampling_hard}). The quantum algorithm can create a superposition of $f^*(s)$ for all $s \in \cS^d_K$ in time $O(\kappa^{d/2} \polylog(N, K))$: in light of the classical lower bounds, this is exponentially faster than any classical algorithm when $\kappa = 1$.\footnote{The classical $\Omega(2^d/\poly(d))$ lower bound holds even when $\kappa=1$.} More generally, our analysis shows that the quantum algorithm has a running time that scales differently from the classical algorithm, i.e., it depends on different parameters such as the condition number, which favors the quantum algorithm in some situations. The quantum advantage comes from exploiting superposition in the multidimensional case: when $d=1$, our quantum algorithm does not provide a quantum speedup. This is because classical LFT computations are efficient in one dimension, but simple ``quantum parallelization'' of a classical algorithm would not be efficient in the multidimensional case, due to the $\Omega(2^d/\poly(d))$ lower bound. We remark that while creating the superposition of $f^*(s)$ for all $s \in \cS^d_K$ can be very efficient, outputting one element of this superposition incurs an extra cost, and limits the quantum advantage to quadratic at most.

We further show that the quantum speedups mentioned in the previous paragraphs are optimal, up to polylogarithmic factors. In particular, our quantum algorithm is optimal for the task of computing the discrete LFT at a specific point (see Corollary~\ref{cor_qdltlowerbound}), and the expected running time of the algorithm features the optimal scaling with respect to a certain parameter (related to the condition number), characterizing for which functions the quantum LFT can be computed efficiently (see Proposition~\ref{cor_qlft_hard} and Corollary~\ref{cor_W_dependence}). As a consequence, our algorithm cannot be substantially improved. 

For tasks that can be implemented as a unitary applied to the superposition of $f^*(s)$ for all $s \in \cS^d_K$, e.g., computing an expectation value of an efficiently-computable observable associated with a Legendre-Fenchel-transformed convex function, the quantum algorithm can provide an exponential speedup compared to any classical algorithm.  Remark~\ref{rmk_expspeedup} discusses possible scenarios for an exponential speedup in more detail. Given that the LFT is a fundamental mathematical tool with applications in various fields, we believe that the quantum algorithm for the LFT presented here may serve as a novel building block in the design of future quantum algorithms. 

\paragraph{Outline}
After discussing some preliminaries in Section~\ref{sec_prelimiaries}, we start by reviewing the classical algorithm for the one-dimensional LFT in Section~\ref{sec_DLFT}, which is then extended to the quantum case in Section~\ref{sec_QLFT_1d}. We distinguish two different cases, depending on choosing a regular or an adaptive discretization on the dual space. In Section~\ref{sec_QLFT_multi} the quantum algorithm is extended to the multidimensional case. 
We discuss the computational complexity of calculating the discrete LFT in Section~\ref{sec_optimality}, showing that no classical algorithm to compute the discrete LFT at a single point can be efficient and that quantum algorithms can be at most quadratically faster for this task. Our quantum algorithm is optimal as it achieves this lower bound (up to polylogarithmic factors).

\section{Preliminaries} \label{sec_prelimiaries}
\subsection{Notation and definitions}
For a vector $x=(x_0,\ldots,x_{N-1})$ its Euclidean norm is denoted by $\norm{x}$.
A function $f:\R^d \to \R$ is called Lipschitz continuous with constant $L\geq 0$ if $|f(x)-f(y)|\leq L \norm{x-y}$ for all $x,y \in \R^d$. 
The function is said to be $\mu$-strongly convex if for all $x,y \in \R^d$ and $t\in [0,1]$ we have $f(tx +(1-t)y)\leq t f(x) +(1-t)f(y) +\frac{1}{2}\mu t(1-t)\norm{x-y}^2$. A function that is differentiable and its gradient is $L'$-Lipschitz continuous is also called $L'$-smooth. The condition number of $f$ is defined by $\kappa:=L'/\mu$~\cite[Section~2.1.3]{ref:nesterov-book-04}. We note that by definition $\kappa \geq 1$. For a multivariate quadratic function $f:[0,1]^d \to \R$ given by $x \mapsto x^{\mathrm{T}}Qx+\langle a,x\rangle +b$ for some positive definite $Q\in \R^{d\times d}$, $a\in \R^d$, and $b \in \R$, the condition number of $f$ coincides with the condition number of the matrix $Q$, i.e., $\kappa=\lambda_{\max}(Q)/\lambda_{\min}(Q)$, where $\lambda_{\max}$ and $\lambda_{\min}$ denote the largest and smallest eigenvalues. 
The indicator function is defined by $\mathds{1}\{X\}:=1$ if $X=\mathrm{true}$ and $0$ otherwise.
The logical 'and' and 'or' operations are denoted by $\wedge$ and $\vee$, respectively. 

\subsection{Regularity assumptions} \label{sec_regularity}
Throughout the entire manuscript we assume the following regularity assumptions of the function $f$:
\begin{assumption}[Regularity assumptions] \label{ass_regularity}
We assume that the function $f:[0,1]^d\to \R$ is:
\begin{enumerate}[(i)]
\item differentiable at the boundary of $[0,1]^d$; \label{it_ass_diff}
\item \label{it_ass_convex} (jointly) convex and $\exists  \nu<\infty$ such that for all $ x_1,\ldots,x_{i-1},x_{i+1},\ldots,x_{d-1} \in[0,1]$ and for all $i\in[d]$
\begin{align*}
    \nabla_{x_i} f(x_0,\ldots,x_{i-1},1,x_{i+1},x_{d-1})-\nabla_{x_i} f(x_0,\ldots,x_{i-1},0,x_{i+1},x_{d-1}) \geq \nu\,;
\end{align*}
\item \label{it_ass_smooth} such that for any $\delta>0$ $\exists \bar L_{\delta}<\infty$ such that for all $ x_i \in[\delta,1-\delta]$, $ x_1,\ldots,x_{i-1},x_{i+1},\ldots,x_{d-1} \in[0,1]$ and for all $i\in[d]$
\begin{align*}
    &f(x_0,\ldots,x_{i-1},x_i+\delta,x_{i+1},x_{d-1})-2f(x_0,\ldots,x_{i-1},x_i,x_{i+1},x_{d-1})\\
    &\hspace{51mm}+ f(x_0,\ldots,x_{i-1},x_i-\delta,x_{i+1},x_{d-1}) \leq \bar L_{\delta} \delta^2\, ;
\end{align*}
\item lower semi-continuous.
\end{enumerate}
\end{assumption}
Assumptions~\eqref{it_ass_diff},~\eqref{it_ass_convex}, and~\eqref{it_ass_smooth} above can be slightly strengthened so that the overall analysis gets less technical and easier to read. However, we emphasize that the stronger assumptions on $f$ (stated below) are not necessary.
\begin{assumption}[Stronger regularity assumptions] \label{ass_regularity_fac}
We assume that the function $f:[0,1]^d\to \R$ is:
\begin{enumerate}
\item[(i+)] differentiable on $[0,1]^d$;
\item[(ii+)] (jointly) $\mu$-strongly convex with $\mu \in \R_+$;
\item[(iii+)] such that $\nabla f$ is $L'$-Lipschitz continuous with $L' \in \R_+$. \label{it_smooth}
\end{enumerate}
\end{assumption}
As suggested by the naming, Assumption~\ref{ass_regularity_fac} implies Assumption~\ref{ass_regularity} for $\nu=\mu$ and $\bar L_{\delta}=L'$.
We recall that the condition number of $f$ is defined as~\cite[Section~2.1.3]{ref:nesterov-book-04} 
\begin{align} \label{eq_condNumber}
  \kappa:=\frac{L'}{\mu} \, .  
\end{align}
In convex optimization, the condition number oftentimes shows up in worst-case running times for several algorithms~\cite{ref:nesterov-book-04}. Note that by definition $\kappa \in [1,\infty)$. A function with a small condition number is called ``well-conditioned", whereas a large condition number indicates that the function is ``ill-conditioned", usually resulting in slower algorithmic performance. As we will see in Section~\ref{sec_QLFT_1d} the running time of the quantum algorithm for computing the LFT of $f$ scales as $\sqrt{\kappa}$.

\subsection{Properties of the Legendre-Fenchel transform} \label{sec_propLFT}
The LFT features several desirable properties. For a detailed overview we refer the interested reader to~\cite[Section~11]{rockafellar2009variational}. Here we discuss only properties that are relevant for our work.
The LFT is an involution for convex and lower semi-continuous functions, i.e., we have $(f^{\ast})^{\ast} = f$.
By construction $f^{\ast}$ is a jointly convex function, irrespective of the function $f$, since it is the pointwise supremum of a family of affine functions~\cite[Section 3.3]{boyd_book}.
It is known~\cite{convexAnalysis93} that $\nabla f$ is $L'$-Lipschitz continuous if, and only if  $f^{\ast}$ is $1/L'$-strongly convex. Because the LFT is an involution for convex functions, this directly implies that the LFT does not change the condition number, i.e., \smash{$\kappa_f = \kappa_{f^{\ast}}$}.
The discrete LFT approximates the continuous LFT when the discretization is sufficiently fine-grained.
\begin{lemma}[{\cite[Lemma~2.2]{QDP20}}] \label{lem_discLFT_contLFT}
    Let $f:\R^d \to \R$ be Lipschitz continuous with constant $L_f$ and let $\cX^d_N$ and $\cS^d_K$ denote the primal and dual space of the discrete LFT. Then,
    \begin{align*}
        |f^{\ast}(s)-f^*(s)| \leq  (1+\sqrt{d}) L_f \mathrm{d_{H}}(\R^d,\cX_N^d) \quad \forall s \in \cS^d_K \, ,
    \end{align*}
    where $\mathrm{d_{H}}(\R^d,\cX_N^d)$ denotes the Hausdorff distance between two sets $\R^d$ and $\cX_N^d$ defined by $\mathrm{d_{H}}(\R^d,\cX_N^d) := \sup_{r \in \R^d} \inf_{x \in \cX_N^d} \norm{r-x}$.
\end{lemma}

\section{One-dimensional classical Legendre-Fenchel transform} \label{sec_DLFT}
Before explaining the classical algorithm~\cite{lucet97} that computes the discrete LFT in time $O(N+K)$, we state basic assumptions on the discretization of the primal and dual spaces. For the dual space we distinguish between a regular discretization, where the grid points are equidistant, and an adaptive discretization, chosen by the algorithm depending the properties of the function to be transformed.

\subsection{Regular discretization} \label{sec_DLFT_regular}
A regular discretization of a one-dimensional continuous space is such that all distances between any two nearby points are equal. Furthermore we assume that the points are sorted. 
\begin{assumption}[One-dimensional regular discretization] \label{ass_discrete}
The discrete sets $\cX^1_N=\{x_0,\ldots,x_{N-1}\}$ and $\cS^1_K=\{s_0,\ldots,s_{K-1}\}$ are such that:
\begin{enumerate}[(i)]
\item The discretization is \emph{sorted}, i.e., $x_{i} \leq x_{i+1}$ for all $i \in [N-1]$ and $s_{j} \leq s_{j+1}$ for all $j \in [K-1]$;
\item The discretization is \emph{regular}, i.e., $x_{i+1} - x_{i} = \gamma_{\mathrm{x}}$ for all $i \in [N-1]$ and $s_{j+1} - s_{j} = \gamma_{\mathrm{s}}$ for all $j \in [K-1]$.
\end{enumerate}
\end{assumption}
Under these assumptions on the discretization and the regularity conditions of the function $f$ mentioned in Assumption~\ref{ass_regularity}, the classical algorithm to compute the discrete LFT presented in~\cite{lucet97} is particularly simple.
Define the discrete gradients
\begin{align} \label{eq_ci}
c_{i}:= \frac{f(x_{i+1})-f(x_i)}{x_{i+1} - x_i} =  \frac{f(x_{i+1})-f(x_i)}{\gamma_{\mathrm{x}}}\, , \qquad \text{for}\quad i \in[N-1] \, ,
\end{align}
and $c_{-1}:=c_0-\eps$, $c_{N-1}:=c_{N-2}+\eps$ for an arbitrarily chosen $\eps >0$.
Because $f$ is convex we have $c_{-1} \leq c_0 \leq \ldots \leq c_{N-1}$. The optimizer $x^\star$ for the discrete LFT~\eqref{eq_DLFT} is then given the following rule: Set $x_0^\star=x_0$, $x^\star_{K-1}=x_{N-1}$ and
\begin{align} \label{eq_LFT_optimizer}
\textnormal{for each } j \in \{1,\ldots,K-2\} \textnormal{ find } i \in [N] \textnormal{ such that } c_{i-1}<s_j \leq c_i  \, \, \Longrightarrow \, \, x^\star_j :=x_i \, .
\end{align}
We refer to~\cite[Lemma~3]{lucet97} for a proof that $x^\star$ defined as above is indeed the correct optimizer. 

With the help of~\eqref{eq_LFT_optimizer} we obtain a linear-time algorithm to compute $f^*(s)$, because sorting two increasing sequences of length $N$ and $K$ takes $O(N+K)$ time. For completeness, the procedure is summarized in Algorithm~\ref{algo_DLFT}.
\begin{algorithm}[!htb]
\caption{One-dimensional discrete LFT with regular discretization~\cite{lucet97}}
\label{algo_DLFT}
\begin{algorithmic}
\STATE \textbf{Input:} $N,K \in \N$, sets $\cX_N^1=\{x_0,\ldots,x_{N-1}\}$ and $\cS_K^1=\{s_0,\ldots,s_{K-1}\}$ satisfying Assumption~\ref{ass_discrete}, and samples $(f(x_0),\ldots,f(x_{n-1}))$, where $f$ is convex; \vspace{2mm} \\
\item \textbf{Output:}  $(f^*(s_0),\ldots,f^*(s_{K-1}))$;\vspace{2mm} \\
Do the following three steps
\begin{enumerate}
\item Compute $(c_{-1},\ldots,c_{N-1})$ with $c_i$ defined in~\eqref{eq_ci};
\item Compute $(x_0^\star,\ldots,x_{K-1}^\star)$ via~\eqref{eq_LFT_optimizer}; 
\item Evaluate $f^*(s_j)=x_j^\star s_j - f(x_j^\star)$ for all $j\in [K]$;
\end{enumerate}
Output $(f^*(s_0),\ldots,f^*(s_{K-1}))$; \vspace{2mm}\\
\end{algorithmic}
\end{algorithm}
\begin{remark}[Range of dual space] \label{rmk_dualSpace}
The nontrivial domain of the discrete LFT is $s_j \in [c_0,c_{N-2}]$, because for all $s_j<c_0$ and $s_j>c_{N-2}$ the corresponding optimizers are always $x^\star_j=x_0$ and $x^\star_j=x_{N-1}$, respectively. In other words, for $s_j \not \in [c_0,c_{N-2}]$ computing $f^*(s_j)$ is trivial.
\end{remark}

\begin{remark}[Sufficient precision] \label{rmk_precision_classical}
Algorithm~\ref{algo_DLFT} assumes sufficient numerical precision so that basic arithmetic operations do not introduce any errors: this assumption makes the analysis of the algorithm significantly simpler. Indeed, numerical errors on the discrete gradients $c_i$ could result in a wrong optimizer computed via the rule~\eqref{eq_LFT_optimizer}, which would then yield an error in the LFT values that depends on $\gamma_{\mathrm{x}}$, due to the discrete nature of the optimizers.
\end{remark}

\begin{remark}[Computing the LFT at a single point] \label{rmk_onedim_classical}
If we are interested in computing the LFT at a single point, rather than at all $s \in \cS_K^1$, there exists a more efficient algorithm than Algorithm~\ref{algo_DLFT}: since $f$ is convex, we can determine the index $i$ in \eqref{eq_LFT_optimizer} via binary search in time $O(\log N)$. This approach is, however, not efficient when $d > 1$, as discussed in Section~\ref{sec_QLFT_multi}.
\end{remark}

The mapping $i \mapsto j$ such that $x_i = x^\star_j$ (as defined in~\eqref{eq_LFT_optimizer}) can be represented by a binary $N \times K$ matrix $M$ defined by $M_{i+1,j+1}=\mathds{1}\{x_i = x^\star_j\}$ for $i\in[N]$ and $j\in[K]$.
In the following we want to quantify how many $s_j$ share the same optimizer. This will be relevant for the quantum algorithm presented in Section~\ref{sec_QLFT_1d}. To do so, we introduce the parameter
\begin{align} \label{eq_W}
W:= \max_{i \in\{1,\ldots,N-2\}} g(i)
=\left \lfloor \max_{i \in\{1,\ldots,N-2\}}\{c_i-c_{i-1}\} \frac{1}{\gamma_{\mathrm{s}}}  \right \rfloor
\leq \left \lfloor \bar L_{\gamma_{\mathrm{x}}} \frac{\gamma_{\mathrm{x}}}{\gamma_{\mathrm{s}}} \right \rfloor\, ,
\end{align}
where we used Assumption~\ref{ass_regularity}~\eqref{it_ass_smooth} to derive the upper bound.
In case $\nabla f$ is $L'$-Lipschitz continuous, we can replace $\bar L_{\gamma_{\mathrm{x}}}$ by $L'$ and hence bound $W$ by
\begin{align} \label{eq_W_bound}
    W \leq L' \frac{\gamma_{\mathrm{x}}}{\gamma_{\mathrm{s}}} =L'\frac{K(x_{N-1}-x_0)}{N(s_{K-1} - s_0)} \, .
\end{align}
It can be seen that each column in $M$ contains exactly a single one, because each $s_j$ has a single optimizer $x^\star_j \in \cX^1_N$. Furthermore, the matrix $M$ contains at most $W$ ones per row and at most $W$ consecutive rows can be entirely zero.
For later use, we also define the parameter
\begin{align}\label{eq_nu}
    \nu:=\frac{c_{N-2}-c_0}{x_{N-1}-x_0} \geq \mu \, ,
\end{align}
where $\mu$ denotes the strong convexity parameter of $f$.


Assuming that the dual space is chosen to include only the nontrivial domain of the discrete LFT (see Remark~\ref{rmk_dualSpace}), i.e. $\cS^1_K=\{c_0,c_0+\gamma_{\mathrm{s}},\ldots,c_{N-2}-\gamma_{\mathrm{s}},c_{N-2}\}=\{s_0,\ldots,s_{K-1}\}$, we define the set
\begin{align} \label{eq_setA}
    &\cA:=\Big \lbrace (i,m) \in [N]\times [W]: \Big(\left \lfloor \frac{c_i - c_{i-1}}{\gamma_{\mathrm{s}}} \right \rfloor \geq  m+1 \wedge i\in \{1,\ldots,N-2\}\Big)  \nonumber \\
    &\hspace{60mm}\vee (m=0 \wedge i=0)\vee (m=0 \wedge i=N-1) \Big \rbrace\,,
\end{align}
which contains, for each $i \in [N]$, an element for each $j\in[K]$ such that $x^\star_j = x_i$. Note that this implies $|\cA| = K$.
Finally, the function
\begin{align} \label{eq_def_j}
 j(i,m,c_{i-1}):=\left \lbrace
 \begin{array}{cl}
    \emptyset  & \textnormal{if } (i,m) \not \in \cA\\
      0 &  \textnormal{if } i=0 \wedge m=0 \\
      k-1 &\textnormal{if } i=N-1 \wedge m=0\\
      \min\limits_{\ell \in [K]}\{\ell+m:\, c_{i-1} < s_{\ell} \} & \textnormal{otherwise}
 \end{array}
 \right. 
\end{align}
is such that $x^\star_{j(i,m,c_{i-1})}=x_i$ and for fixed discrete gradients $c_{-1},c_0,\ldots,c_{N-2}$ we have $\{j(i,m,c_{i-1}):i\in[N],m \in [W]\}=[K]$.

\begin{remark}[Extension to nonconvex functions] \label{rmk_nonConvex}
The discrete LFT of a nonconvex function $f$ can be computed by determining the convex hull of $f$ before starting Algorithm~\ref{algo_DLFT}~\cite{lucet97}. The convex hull of $N$ points in $\R^2$ can be computed in time $O(N \log h)$, where $h$ is the number of points comprising the hull, which is known to be optimal~\cite{seidel86}. We note that this additional convex hull step will not affect the overall complexity for the computation of the discrete LFT.\footnote{This is not the case for the quantum LFT discussed in the next section: it is not clear how to extend the algorithm to nonconvex functions without increasing the running time. We refer to Remark~\ref{rmk_quatum_nonConvex} for a detailed explanation.} 
\end{remark}

\subsection{Adaptive discretization} \label{sec_DLFT_adaptive}
We consider the same regular discretization of the primal space as above, i.e., the set $\cX_N^1$ satisfies Assumption~\ref{ass_discrete}. The idea of an adaptive discretization in the dual space is to choose $K=N$ and 
\begin{align} \label{eq_adaptive_dual_set}
\cS_{N,\mathrm{adaptive}}^1=\left\lbrace\frac{c_{-1}+c_0}{2},\frac{c_0+c_1}{2},\ldots,\frac{c_{N-3}+c_{N-2}}{2},\frac{c_{N-2}+c_{N-1}}{2}\right \rbrace \, .    
\end{align}
This choice implies, that according to~\eqref{eq_LFT_optimizer}, the optimizer in~\eqref{eq_DLFT} satisfies $x^\star_i=x_i$. We have a unique optimizer of each discretization point in the dual space: this is the major difference between adaptive discretization and regular discretization.
The classical algorithm for the adaptive discretization is summarized in Algorithm~\ref{algo_DLFT_adaptive}. It is straightforward to see that its running time scales as $O(N)$.
The importance of the adaptive discretization may not be clear in the classical setting discussed here. However, its relevance will be justified in Section~\ref{sec_QLFT_adaptive}, where we discuss a quantum algorithm to compute the LFT: in that case the adaptive discretization may lead to a better running time compared to the regular discretization.

\begin{algorithm}[!htb]
\caption{One-dimensional discrete LFT with adaptive discretization}
\label{algo_DLFT_adaptive}
\begin{algorithmic}
\STATE \textbf{Input:} $N \in \N$, $\cX_N^1=\{x_0,\ldots,x_{N-1}\}$ satisfying Assumption~\ref{ass_discrete}, and samples $(f(x_0),\ldots,f(x_{n-1}))$, where $f$ is convex; \vspace{2mm} \\
\item \textbf{Output:}  $(f^*(s_0),\ldots,f^*(s_{K-1}))$ for $\cS_{N,\mathrm{adaptive}}^1=\{\frac{c_{-1}+c_0}{2},\frac{c_0+c_1}{2},\ldots,\frac{c_{N-3}+c_{N-2}}{2},\frac{c_{N-2}+c_{N-1}}{2}\}$ with $c_i$ defined in~\eqref{eq_ci};  \vspace{2mm} \\
Do the following three steps
\begin{enumerate}
\item Compute $(c_{-1},\ldots,c_{N-1})$ with $c_i$ defined in~\eqref{eq_ci};
\item Let $(s_0,s_1,\ldots,s_{N-2},s_{N-1})=(\frac{c_{-1}+c_0}{2},\frac{c_0+c_1}{2},\ldots,\frac{c_{N-3}+c_{N-2}}{2},\frac{c_{N-2}+c_{N-1}}{2})$; 
\item Evaluate $f^*(s_i)=x_i s_i - f(x_i)$ for all $i\in[N]$;
\end{enumerate}
Output $(f^*(s_0),\ldots,f^*(s_{N-1}))$; \vspace{2mm}\\
\end{algorithmic}
\end{algorithm}

\begin{remark}\label{rmk_adaptive_disc}
There are other possibilities than~\eqref{eq_adaptive_dual_set} to define the adaptive dual set, such as, for example
\begin{align*}
\cS_{K,\mathrm{adaptive,right}}^1=\{c_0,c_1,\ldots,c_{N-2},c_{N-2}\} \qquad \textnormal{or} \qquad
\cS_{K,\mathrm{adaptive,left}}^1=\{c_0,c_0,c_1,\ldots,c_{N-2}\} \, .    
\end{align*}
The choice $\cS_{K,\mathrm{adaptive}}$ has the advantage that for piecewise linear functions, the dual points are (intuitively) potentially located  ``in the middle of the linear pieces'',  which can provide useful information about the shape of $f^\star$. This is illustrated in the examples in Section~\ref{sec_examples} next.
\end{remark}

\subsection{Examples} \label{sec_examples}
In this section we discuss three examples that illustrate the LFT algorithm, and in particular the differences between regular and adaptive discretization. We discuss one quadratic function and two different piecewise linear functions. These examples will be useful also in the next section, where we develop a quantum algorithm for the LFT. We will see that the success probability of the quantum algorithm depends on the intrinsic structure of the function. 

\begin{example}[Quadratic function] \label{ex_quadratic}
Consider the convex quadratic function $f:[0,1] \to \R$ defined by 
\begin{align*}
f(x)= x^2 - \frac{3}{4}x + \frac{1}{2} \, .
\end{align*}
This function satisfies Assumptions~\ref{ass_regularity} and~\ref{ass_regularity_fac}; more precisely, it is $2$-strongly convex, its gradient is $2$-Lipschitz continuous and hence $\kappa=1$.
The continuous LFT can be computed analytically for this example as 
\begin{align*}
f^{\ast}(s) = \left \lbrace 
\begin{array}{cc}
    -1/2 & s<-3/4  \\
    s^2/4+3s/8-23/64 & -3/4 \leq s \leq 5/4 \\
    1/2 & s>5/4 \, .
\end{array} \right .
\end{align*}
For a regular discretization $\cX^1_5=\{0,1/4,1/2,3/4,1\}$, we compute $(c_{-1},\ldots,c_4)=(-1/2-\eps,-1/2,0,$ $1/2,1,1+\eps)$, where $\eps>0$ is an arbitrary constant.
We thus see that the nontrivial range of the dual space is $[c_0,c_{3}]=[-1/2,1]$ as explained in Remark~\ref{rmk_dualSpace}. We may choose a regular dual set $\cS^1_4=\{-1/2,0,1/2,1\}$ of size $4$. Algorithm~\ref{algo_DLFT} then gives $(x_0^\star,x_1^\star,x_2^\star,x_3^\star)=(x_0,x_1,x_2,x_3)=(0,1/4,1/2,1)$ and finally $(f^*(s_0),f^*(s_1),f^*(s_2),f^*(s_3))=(-1/2,-3/8,-1/8,1/4)$, as illustrated in Figure~\ref{fig_example_LFT}. Note that the discrete LFT does not coincide with the continuous LFT, but their values are not too different (see Lemma~\ref{lem_discLFT_contLFT}). The adaptive Algorithm~\ref{algo_DLFT_adaptive} outputs $(f^*(s_0),f^*(s_1),f^*(s_2),f^*(s_3))=(-1/2,-7/16,-1/4,$ $1/16,1/4)$ for $\cS_{5,\mathrm{adaptive}}^1=\{-1/2,-1/4,1/4,3/4,1\}$, and $(f^*(s_0),f^*(s_1),f^*(s_2),f^*(s_3))=(-1/2,$ $-3/8,-1/8,1/4,1/4)$ for $\cS_{5,\mathrm{adaptive,right}}^1=\{-1/2,0,1/2,1,1\}$.
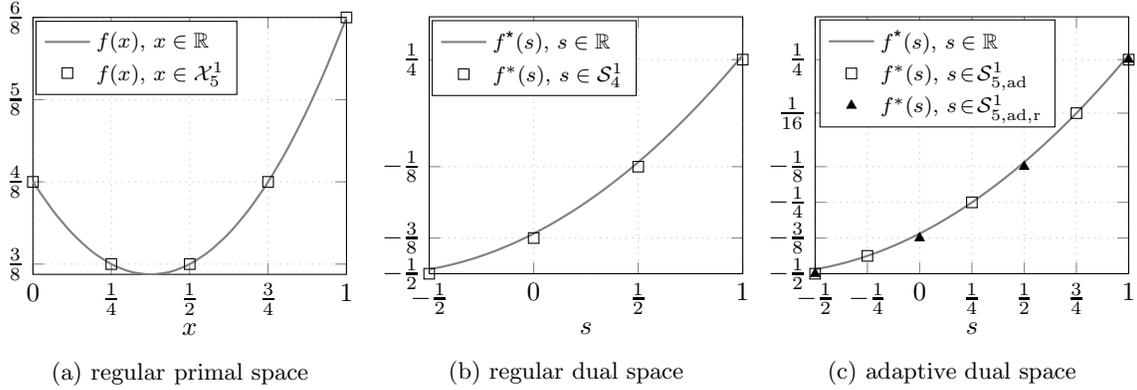
\begin{figure}[!htb]
\centering
\begin{minipage}[c]{.33333\textwidth}
\centering
\vspace{-3mm}
  \begin{tikzpicture}
	\begin{axis}[
		height=5cm,
		width=5.75cm,
		grid=major,
	    grid style=dotted,
		xlabel=$x$,
		xmin=0,
		xmax=1,
		ymax=0.75,
		ymin=0.359375,
	     xtick={0,0.25,0.5,0.75,1},
	     xticklabels={$0$, $\frac{1}{4}$, $\frac{1}{2}$, $\frac{3}{4}$, $1$},
         ytick={0.75,0.625,0.5,0.375},
	     yticklabels={$\frac{6}{8}$,$\frac{5}{8}$, $\frac{4}{8}$, $\frac{3}{8}$},          
		legend style={at={(0.33,0.987)},anchor=north,legend cell align=left,font=\footnotesize} 
	]

	\addplot[gray,thick,smooth] coordinates {
(0, 1/2) (1/16, 117/256) (1/8, 27/64) (3/16, 101/256) (1/4, 3/8) (5/16, 93/256) (3/8, 23/64) (7/16, 93/256) (1/2,3/8) (9/16,101/256) (5/8, 27/64) (11/16, 117/256) (3/4, 1/2) (13/16, 141/256) (7/8, 39/64) (15/16, 173/256) (1, 3/4)
	};
	\addlegendentry{$f(x)$, $x \in \R$}
			
	\addplot[black,mark=square,only marks] coordinates {
(0, 1/2) (1/4, 3/8) (1/2, 3/8) (3/4, 1/2) (1, 3/4)
	};
	\addlegendentry{$f(x)$, $x \in \cX^1_5$}
	\end{axis}   
\end{tikzpicture}	
        \subcaption{regular primal space}
\end{minipage}
	\hspace{-3mm} 	
  \begin{minipage}[c]{.33333\textwidth}
        \centering
	  \begin{tikzpicture}
	\begin{axis}[
		height=5cm,
		width=5.75cm,
	    grid style=dotted,
		grid=major,
		xlabel=$s$,
		xmin=-0.5,
		xmax=1,
		ymax=0.4,
		ymin=-0.5,
	     xtick={-0.5,0,0.5,1},
	     xticklabels={$-\frac{1}{2}$, $0$, $\frac{1}{2}$, $1$},
         ytick={0.25,-0.125,-0.375,-0.5},
    	 yticklabels={$\frac{1}{4}$, $-\frac{1}{8}$,$-\frac{3}{8}$ ,$-\frac{1}{2}$},      
		legend style={at={(0.335,0.987)},anchor=north,legend cell align=left,font=\footnotesize} 
	]

	\addplot[gray,thick,smooth] coordinates {
(-1/2), -31/64) (-7/16), -487/1024) (-3/8), -119/256) (-5/16), -463/1024) (-1/4), -7/16) (-3/16), -431/1024) (-1/8), -103/256) (-1/16), -391/1024) (0, -23/64) (1/16, -343/1024) (1/8, -79/256) (3/16, -287/1024) (1/4, -1/4) (5/16, -223/1024) (3/8, -47/256) (7/16, -151/1024) (1/2, -7/64) (9/16, -71/1024) (5/8, -7/256) (11/16, 17/1024) (3/4, 1/16) (13/16, 113/1024) (7/8, 41/256) (15/16, 217/1024) (1, 17/64)
	};
	\addlegendentry{$f^{\ast}(s)$, $s \in \R$}
			
	\addplot[black,mark=square,only marks] coordinates {
(-1/2,-1/2) (0,-3/8) (1/2,-1/8) (1,1/4)
	};
	\addlegendentry{$f^*(s)$, $s \in \cS^1_4$}
                                			
	\end{axis}  
\end{tikzpicture}
        \subcaption{regular dual space}
\end{minipage}
	\hspace{-3mm} 
  \begin{minipage}[c]{.33333\textwidth}
        \centering
	  \begin{tikzpicture}
	\begin{axis}[
		height=5cm,
		width=5.75cm,
	    grid style=dotted,
		grid=major,
		xlabel=$s$,
		xmin=-0.5,
		xmax=1,
		ymax=0.4,
		ymin=-0.5,
	     xtick={-0.5,-0.25,0,0.25,0.5,0.75,1},
	     xticklabels={$-\frac{1}{2}$, $-\frac{1}{4}$, $0$, $\frac{1}{4}$, $\frac{1}{2}$,$\frac{3}{4}$,$1$},
         ytick={0.25,0.0625,-0.125,-0.25,-0.375,-0.5},
    	 yticklabels={$\frac{1}{4}$,$\frac{1}{16}$, $-\frac{1}{8}$,$-\frac{1}{4}$,$-\frac{3}{8}$ ,$-\frac{1}{2}$},      
		legend style={at={(0.39,0.99)},anchor=north,legend cell align=left,font=\footnotesize} 
	]

	\addplot[gray,thick,smooth] coordinates {
(-1/2), -31/64) (-7/16), -487/1024) (-3/8), -119/256) (-5/16), -463/1024) (-1/4), -7/16) (-3/16), -431/1024) (-1/8), -103/256) (-1/16), -391/1024) (0, -23/64) (1/16, -343/1024) (1/8, -79/256) (3/16, -287/1024) (1/4, -1/4) (5/16, -223/1024) (3/8, -47/256) (7/16, -151/1024) (1/2, -7/64) (9/16, -71/1024) (5/8, -7/256) (11/16, 17/1024) (3/4, 1/16) (13/16, 113/1024) (7/8, 41/256) (15/16, 217/1024) (1, 17/64)
	};
	\addlegendentry{$f^{\ast}(s)$, $s \in \R$}
			
	\addplot[black,mark=square,only marks] coordinates {
(-1/2,-1/2) (-1/4,-7/16) (1/4,-1/4) (3/4,1/16) (1,1/4)
	};
	\addlegendentry{$f^*(s)$, $s \!\in\! \cS^1_{5,\mathrm{ad}}$}
	\addplot[black,mark=triangle*,only marks] coordinates {
(-1/2,-1/2) (0,-3/8) (1/2,-1/8) (1,1/4)
	};
	\addlegendentry{$f^*(s)$, $s \!\in\! \cS^1_{5,\mathrm{ad,r}}$}

	\end{axis}  
\end{tikzpicture}
        \subcaption{adaptive dual space}
  \end{minipage}
\caption{Graphical visualization of Example~\ref{ex_quadratic} for the discrete sets (a) $\cX^1_5=\{0,1/4,1/2,3/4,1\}$, (b) $\cS^1_4=\{-1/2,0,1/2,1\}$, (c) $\cS_{5,\mathrm{adaptive}}^1=\{-1/2,-1/4,1/4,3/4,1\}$, and (c) $\cS^1_{5,\mathrm{adaptive,right}}=\{-1/2,0,1/2,1,1\}$. The dual space is plotted in the nontrivial domain $[-1/2,1]$, as explained in Remark~\ref{rmk_dualSpace}.}
\label{fig_example_LFT}
\end{figure}
\end{example}

\begin{example}[Piecewise linear function] \label{ex_piecewise_linear_1}
Consider the convex and piecewise linear function $f:[0,1] \to \R$ defined by 
\begin{align*}
    f(x) = \left \lbrace \begin{array}{cc}
        0 & 0\leq x < 1/4\\
        x/4 - 1/16 & 1/4 \leq x < 1/2 \\
        x/2 - 3/16 & 1/3 \leq x < 3/4 \\
        3x/4 - 6/16 & 3/4 \leq x \leq 1\, . 
    \end{array} \right .
\end{align*}
This function satisfies Assumption~\ref{ass_regularity} with $\nu=3/4$ and $\bar L_{1/4}=1$; it is depicted in Figure~\ref{fig_example_piecewiseLinear}. We consider a regular discretization $\cX^1_5=\{0,1/4,1/2,3/4,1\}$. Its continuous LFT can be computed analytically as
\begin{align*}
    f^{\ast}(s) = \left \lbrace \begin{array}{cc}
        0 &  s< 0\\
        s/4  & 0 \leq s < 1/4 \\
        s/2 - 1/16 & 1/4 \leq s < 1/2 \\
        3s/4 - 3/16 & 1/2 \leq s < 3/4\\
        s-6/16 & s\geq 3/4.
    \end{array} \right .
\end{align*}
To study the discrete case we compute $(c_{-1},\dots,c_4)=(0-\eps,0,1/4,1/2,3/4,3/4+\eps)$ for an arbitrary constant $\eps >0$. We see that the nontrivial range of the dual space is $[c_0,c_{3}]=[0,3/4]$, hence we may choose the regular dual set $\cS^1_5=\{0,3/16,6/16,9/16,3/4\}$. We next use Algorithm~\ref{algo_DLFT} to compute $(x_0^\star,x_1^\star,x_2^\star,x_3^\star,x_4^\star)=(x_0,x_1,x_2,x_3,x_4)=(0,1/4,1/2,3/4,1)$ and $(f^*(s_0),f^*(s_1),f^*(s_2),f^*(s_3),f^*(s_4))$ $=(0,3/64,1/8,15/64,6/16)$, as illustrated in Figure~\ref{fig_example_piecewiseLinear}.\footnote{We note that the parameter $W=\lfloor(1/4)/(3/16) \rfloor=\lfloor 4/3\rfloor=1$, as defined in~\eqref{eq_W}, hence
$K/(NW)=1$. As a result, the quantum Algorithm presented in Section~\ref{sec_QLFT_1d} succeeds with probability $1$.}
The adaptive Algorithm~\ref{algo_DLFT_adaptive} outputs the vector $(f^*(s_0),f^*(s_1),f^*(s_2),f^*(s_3),f^*(s_4))\!=\!(0,1/32,1/8,9/32,6/16)$ for $\cS_{5,\mathrm{adaptive}}^1\!=\!\{0,1/8,3/8,5/8,6/8\}$, and $(f^*(s_0),f^*(s_1),f^*(s_2),f^*(s_3),f^*(s_4))\!=\!(0,1/16,3/16,6/16,6/16)$ for $\cS_{5,\mathrm{adaptive,right}}^1\!=\!\{0,1/4,1/2,$ $3/4,3/4\}$.
\begin{figure}[!htb]
\centering
  \begin{minipage}[c]{.33333\textwidth}
        \centering
            \vspace{-3mm}
  \begin{tikzpicture}
	\begin{axis}[
		height=5cm,
		width=5.75cm,
		 grid style=dotted,
	    grid=major,
		xlabel=$x$,
		xmin=0,
		xmax=1,
		ymax=0.375,
		ymin=0,
	     xtick={0,0.25,0.5,0.75,1},
	     xticklabels={$0$, $\frac{1}{4}$, $\frac{1}{2}$, $\frac{3}{4}$, $1$},
         ytick={0.375,0.1875,0.0625,0},	     
         yticklabels={$\frac{6}{16}$,$\frac{3}{16}$,$\frac{1}{16}$,$0$},
		legend style={at={(0.33,0.987)},anchor=north,legend cell align=left,font=\footnotesize} 
	]

	\addplot[gray,thick] coordinates {
(0,0) (1/4,0) (1/2,1/16) (3/4,3/16) (1,6/16)
	};
	\addlegendentry{$f(x)$, $x \in \R$}
			
	\addplot[black,mark=square,only marks] coordinates {
(0,0) (1/4,0) (1/2,1/16) (3/4,3/16) (1,6/16)
	};
	\addlegendentry{$f(x)$, $x \in \cX^1_5$}
                                			
	\end{axis}   
\end{tikzpicture}
\subcaption{regular primal space}
\end{minipage}
	\hspace{-3mm} 
  \begin{minipage}[c]{.33333\textwidth}
        \centering	
	  \begin{tikzpicture}
	\begin{axis}[
		height=5cm,
		width=5.75cm,
		 grid style=dotted,
	    grid=major,
		xlabel=$s$,
		xmin=0,
		xmax=0.75,
		ymax=0.5,
		ymin=0,
	     xtick={0,0.1875,0.25,0.375,0.5,0.5625,0.75},
	     xticklabels={$0$,$\frac{3}{16}$,$\frac{1}{4}$,$\frac{6}{16}$,$\frac{1}{2}$,$\frac{9}{16}$,$\frac{3}{4}$},
         ytick={0.375,0.1875,0.0625,0},	     
         yticklabels={$\frac{6}{16}$,$\frac{3}{16}$,$\frac{1}{16}$,$0$},
		legend style={at={(0.34,0.987)},anchor=north,legend cell align=left,font=\footnotesize} 
	]

	\addplot[gray,thick] coordinates {
(0,0) (1/4,1/16) (1/2,3/16) (3/4,6/16)
	};
	\addlegendentry{$f^{\ast}(s)$, $s \in \R$}
			
	\addplot[black,mark=square,only marks] coordinates {
(0,0) (3/16,3/64) (6/16,1/8) (9/16,15/64) (12/16,6/16)
	};
	\addlegendentry{$f^*(s)$, $s \in \cS^1_5$}
	\end{axis}  
\end{tikzpicture}
\subcaption{regular dual space}
\end{minipage}
	\hspace{-3mm} 
  \begin{minipage}[c]{.33333\textwidth}
        \centering	
	  \begin{tikzpicture}
	\begin{axis}[
		height=5cm,
		width=5.75cm,
		 grid style=dotted,
	    grid=major,
		xlabel=$s$,
		xmin=0,
		xmax=0.75,
		ymax=0.5,
		ymin=0,
	     xtick={0,0.125,0.25,0.375,0.5,0.625,0.75},
	     xticklabels={$0$,$\frac{1}{8}$,$\frac{1}{4}$,$\frac{3}{8}$,$\frac{1}{2}$,$\frac{5}{8}$,$\frac{3}{4}$},
         ytick={0.375,0.28125,0.1875,0.125,0.0625,0},	     
         yticklabels={$\frac{6}{16}$,$\frac{9}{32}$,$\frac{3}{16}$,$\frac{1}{8}$,$\frac{1}{16}$,$0$},
		legend style={at={(0.4,0.99)},anchor=north,legend cell align=left,font=\footnotesize} 
	]

	\addplot[gray,thick] coordinates {
(0,0) (1/4,1/16) (1/2,3/16) (3/4,6/16)
	};
	\addlegendentry{$f^{\ast}(s)$, $s \in \R$}
			
	\addplot[black,mark=square,only marks] coordinates {
(0,0) (1/8,1/32) (3/8,1/8) (5/8,9/32) (3/4,6/16)
	};
	\addlegendentry{$f^*(s)$, $s \in \cS^1_{5,\mathrm{ad}}$}
	
 	\addplot[black,mark=triangle*,only marks] coordinates {
(0,0) (1/4,1/16) (1/2,3/16) (3/4,6/16)
	};
	\addlegendentry{$f^*(s)$, $s \in \cS^1_{5,\mathrm{ad,r}}$}

	\end{axis}  
\end{tikzpicture}
\subcaption{adaptive dual space}
\end{minipage}
\caption{Graphical visualization of Example~\ref{ex_piecewise_linear_1} for the discrete sets (a) $\cX^1_5=\{0,1/4,1/2,3/4,1\}$, (b) $\cS^1_5=\{0,3/16,6/16,9/16,3/4\}$, (c) $\cS_{5,\mathrm{adaptive}}^1=\{0,1/8,3/8,5/8,6/8\}$, and (c) $\cS^1_{5,\mathrm{adaptive,right}}=\{0,1/4,1/2,3/4,3/4\}$. We see that the adaptive (centered) discretization leads to points in the middle of the linear segments. The dual space is plotted in the nontrivial domain $[0,3/4]$, as explained in Remark~\ref{rmk_dualSpace}.}
\label{fig_example_piecewiseLinear}
\end{figure}
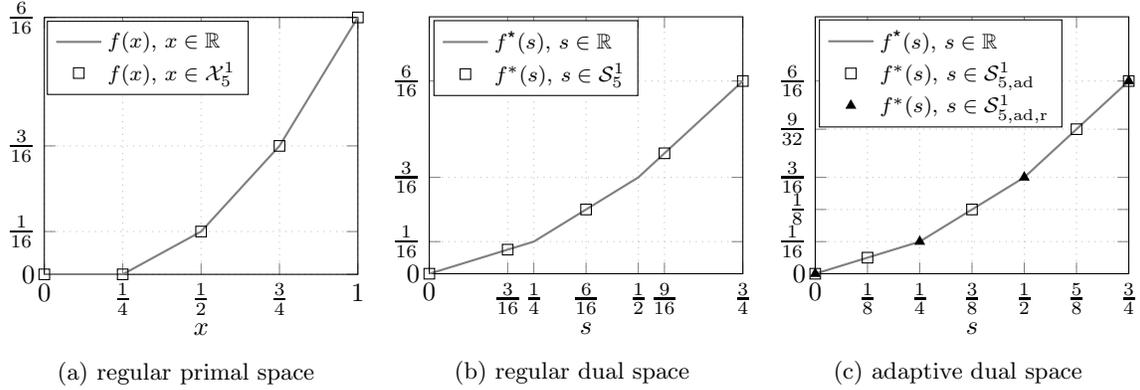
\end{example}
\begin{example}[Piecewise linear function] \label{ex_piecewise_linear_2}
Consider the convex and piecewise linear function $f:[0,1] \to \R$ defined by 
\begin{align*}
    f(x) = \left \lbrace \begin{array}{cl}
        0 & 0\leq x < 1/4\\
        x/2 - 1/8 & 1/4 \leq x < 3/4 \\
        x - 1/2 & 3/4 \leq x \leq 1 \, .
    \end{array} \right .
\end{align*}
This function satisfies Assumption~\ref{ass_regularity} with $\nu=1$ and $\bar L_{1/4}=2$; it is depicted in Figure~\ref{fig_example_piecewiseLinear_2}. Let $\cX^1_5=\{0,1/4,1/2,3/4,1\}$ be a regular discretization of the primal space. The continuous LFT of $f$ can be computed analytically as
\begin{align*}
    f^{\ast}(s) = \left \lbrace \begin{array}{cl}
        0 &  s< 0\\
        s/4  & 0 \leq s < 1/2 \\
        3s/2 - 1/4 & 1/2 \leq s < 1 \\
        s-1/2 & s\geq 1.
    \end{array} \right .
\end{align*}
To study the discrete case, we compute $(c_{-1},\dots,c_4)=(0-\eps,0,1/2,1/2,1,1+\eps)$ for an arbitrary constant $\eps >0$. We see that the nontrivial range of the dual space is $[c_0,c_{3}]=[0,1]$; hence, we may choose $\cS^1_5=\{0,1/4,1/2,3/4,1\}$ as a regular dual set. We next use Algorithm~\ref{algo_DLFT} to compute $(x_0^\star,x_1^\star,x_2^\star,x_3^\star,x_4^\star)=(x_0,x_1,x_1,x_3,x_4)=(0,1/4,1/4,3/4,1)$ and $(f^*(s_0),f^*(s_1),f^*(s_2),f^*(s_3),f^*(s_4))$ $=\!(0,1/16,1/8,5/16,1/2)$, as illustrated in Figure~\ref{fig_example_piecewiseLinear}.\footnote{We note that the parameter $W=\lfloor(1/2)/(1/4) \rfloor=2$, as defined in~\eqref{eq_W}, hence $K/(NW)= 1/2$. This implies that the quantum Algorithm presented in Section~\ref{sec_QLFT_1d} succeeds with probability $1/2$.}
The adaptive Algorithm~\ref{algo_DLFT_adaptive} outputs $(f^*(s_0),f^*(s_1),$ $f^*(s_2),f^*(s_3),f^*(s_4))=(0,1/16,1/8,5/16,1/2)$ for $\cS_{5,\mathrm{adaptive}}^1=\{0,1/4,1/2,3/4,1\}$, and $(f^*(s_0),$ $f^*(s_1),f^*(s_2),f^*(s_3),f^*(s_4))=(0,1/8,1/8,1/2,1/2)$ for $\cS_{5,\mathrm{adaptive,right}}^1=\{0,1/2,1/2,1,1\}$.
\begin{figure}[!htb]
\centering
\begin{minipage}[c]{.33333\textwidth}
\vspace{-3mm}
\centering
  \begin{tikzpicture}
	\begin{axis}[
		height=5cm,
		width=5.75cm,
		 grid style=dotted,
	    grid=major,
		xlabel=$x$,
		xmin=0,
		xmax=1,
		ymax=0.5,
		ymin=0,
	     xtick={0,0.25,0.5,0.75,1},
	     xticklabels={$0$, $\frac{1}{4}$, $\frac{1}{2}$, $\frac{3}{4}$, $1$},
         ytick={0.5,0.25,0.125,0},	     
         yticklabels={$\frac{1}{2}$,$\frac{1}{4}$,$\frac{1}{8}$,$0$},
		legend style={at={(0.33,0.987)},anchor=north,legend cell align=left,font=\footnotesize} 
	]

	\addplot[gray,thick] coordinates {
(0,0) (1/4,0) (3/4,1/4) (1,1/2)
	};
	\addlegendentry{$f(x)$, $x \in \R$}
			
	\addplot[black,mark=square,only marks] coordinates {
(0,0) (1/4,0) (1/2,1/8) (3/4,1/4) (1,1/2)
	};
	\addlegendentry{$f(x)$, $x \in \cX^1_5$}
                                			
	\end{axis}   
\end{tikzpicture}
\subcaption{regular primal space}
\end{minipage}
	\hspace{-3mm} 
  \begin{minipage}[c]{.33333\textwidth}	
	  \begin{tikzpicture}
	\begin{axis}[
		height=5cm,
		width=5.75cm,
		 grid style=dotted,
	    grid=major,
		xlabel=$s$,
		xmin=0,
		xmax=1,
		ymax=0.65,
		ymin=0,
	     xtick={0,0.25,0.5,0.75,1},
	     xticklabels={$0$,$\frac{1}{4}$,$\frac{1}{2}$,$\frac{3}{4}$,$1$},
         ytick={0.5,0.3125,0.125,0.065,0},	     
         yticklabels={$\frac{1}{2}$,$\frac{5}{16}$,$\frac{1}{8}$,$\frac{1}{16}$,$0$},
		legend style={at={(0.34,0.987)},anchor=north,legend cell align=left,font=\footnotesize} 
	]

	\addplot[gray,thick] coordinates {
(0,0) (1/2,1/8) (1,1/2)
	};
	\addlegendentry{$f^{\ast}(s)$, $s \in \R$}
			
	\addplot[black,mark=square,only marks] coordinates {
(0,0) (1/4,1/16) (1/2,1/8) (3/4,5/16) (1,1/2)
	};
	\addlegendentry{$f^*(s)$, $s \in \cS^1_5$}

	\end{axis}  
\end{tikzpicture}
\subcaption{regular dual space}
\end{minipage}
	\hspace{-3mm} 
  \begin{minipage}[c]{.33333\textwidth}	
	  \begin{tikzpicture}
	\begin{axis}[
		height=5cm,
		width=5.75cm,
		 grid style=dotted,
	    grid=major,
		xlabel=$s$,
		xmin=0,
		xmax=1,
		ymax=0.65,
		ymin=0,
	     xtick={0,0.25,0.5,0.75,1},
	     xticklabels={$0$,$\frac{1}{4}$,$\frac{1}{2}$,$\frac{3}{4}$,$1$},
         ytick={0.5,0.3125,0.125,0.0625,0},	     
         yticklabels={$\frac{1}{2}$,$\frac{5}{16}$,$\frac{1}{8}$,$\frac{1}{16}$,$0$},
		legend style={at={(0.4,0.99)},anchor=north,legend cell align=left,font=\footnotesize} 
	]

	\addplot[gray,thick] coordinates {
(0,0) (1/2,1/8) (1,1/2)
	};
	\addlegendentry{$f^{\ast}(s)$, $s \in \R$}
			
	\addplot[black,mark=square,only marks] coordinates {
(0,0) (1/4,1/16)  (1/2,1/8) (3/4,5/16) (1,1/2)
	};
	\addlegendentry{$f^*(s)$, $s \in \cS^1_{5,\mathrm{ad}}$}
	
	\addplot[black,mark=triangle*,only marks] coordinates {
(0,0)  (1/2,1/8) (1,1/2)
	};
	\addlegendentry{$f^*(s)$, $s \in \cS^1_{5,\mathrm{ad,r}}$}

	\end{axis}  
\end{tikzpicture}
\subcaption{adaptive dual space}
\end{minipage}
\caption{Graphical visualization of Example~\ref{ex_piecewise_linear_2} for the discrete sets (a) $\cX^1_5=\{0,1/4,1/2,3/4,1\}$, (b) $\cS^1_5=\{0,1/4,1/2,3/4,1\}$, (c) $\cS_{5,\mathrm{adaptive}}^1=\{0,1/4,1/2,3/4,1\}$, and (c) $\cS^1_{5,\mathrm{adaptive,right}}=\{0,1/2,1/2,1,1\}$. The dual space is plotted in the nontrivial domain $[0,1]$, as explained in Remark~\ref{rmk_dualSpace}.}
\label{fig_example_piecewiseLinear_2}
\end{figure}
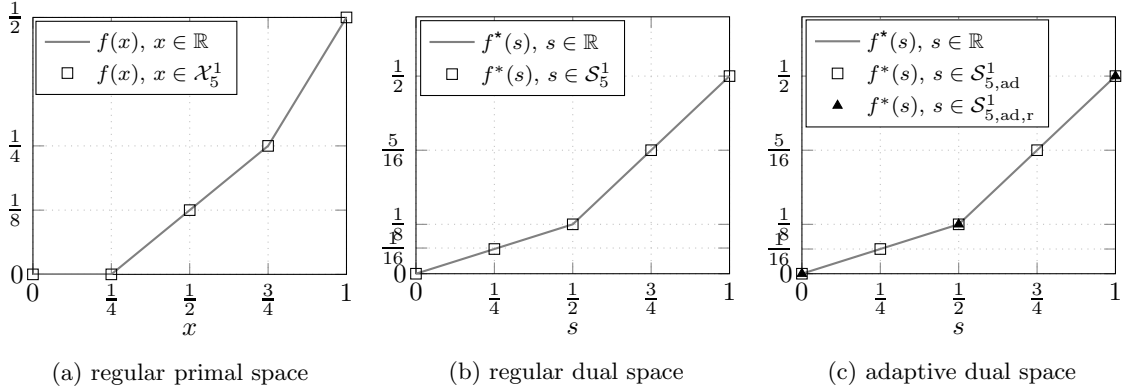
\end{example}
\section{One-dimensional quantum Legendre-Fenchel transform} \label{sec_QLFT_1d}
Let $f:[0,1]\to \R$ satisfy Assumption~\ref{ass_regularity} and consider discrete sets $\cX^1_N=\{x_0,\ldots,x_{N-1}\} \subseteq [0,1]$ and $\cS^1_K=\{s_0,\ldots,s_{K-1}\}\subseteq \R$. The \emph{quantum Legendre-Fenchel transform} (QLFT) is defined by the mapping
\begin{align} \label{eq_QLFT}
 \frac{1}{\sqrt{N}} \sum_{i=0}^{N-1} \ket{i} \ket{f(x_i)} \quad \to \quad \frac{1}{\sqrt{K}} \sum_{j=0}^{K-1} \ket{j} \ket{f^*(s_j)}\, .
\end{align} 
\begin{remark}[Connections to Fourier transform] \label{rmk_FT}
Unlike the Fourier transform, the LFT is a nonlinear operation. In the \emph{max-plus algebra} with the semiring $\R_{\max}=[\R \cup \{-\infty\},\oplus,\otimes]$, where $a\oplus b:= \max\{a,b\} $ and $ a \otimes b :=   a+b$, the LFT is a linear operation that corresponds to the Fourier transform in the standard algebra~\cite{lucet09}.
\end{remark}
Given that the discrete Fourier transform features an efficient quantum-mechanical implementation~\cite{QFT94,nielsenChuang_book}, Remark~\ref{rmk_FT} above raises hope that the same happens for the discrete LFT. This is indeed the case, and we show that there exists a quantum algorithm for the QLFT that can be exponentially faster than any classical discrete LFT algorithm. We achieve this by generalizing Algorithms~\ref{algo_DLFT} and~\ref{algo_DLFT_adaptive} to the quantum case. The main intuition behind the quantum algorithms is to start with the superposition on the left-hand side of \eqref{eq_QLFT}, and compute the right-hand side via relabeling. Indeed, suppose for simplicity that $N = K$, and that the function is sufficiently well-behaved that each slope $s_j$ has a distinct optimizer $x_i$, according to \eqref{eq_LFT_optimizer}. In this case, the function \eqref{eq_QLFT} can be constructed by mapping each index $i$ to the index $j$ such that $x_i$ is the optimizer for $f^*(s_j)$ and then computing $f^*(s_j)$ according to \eqref{eq_DLFT}. In reality, the construction is more involved, because multiple $s_j$ may have the same optimizer $x_i$, and we need an efficient algorithm to compute how each index $i$ maps to $j$ such that $x_i$ is the optimizer for $s_j$. However, this high-level explanation provides the intuition for the quantum advantage: we start from a superposition in primal space, and build a superposition in dual space, using an algorithm that is efficient if the condition number of $f$ is $1$ or close to 1. This avoids classical computation of all the discrete gradients, which would require at least $N$ operations.

\subsection{Input and output model}
\label{sec_input_output}
For $N=2^n$ with $n\in \N$, the analog encoding of a vector $x \in \R^N$ into an $n$-qubit state $\frac{1}{\norm{x}}\sum_{i=0}^{N-1} x_i \ket{i}$ is the following unitary transformation $U_A$:
\begin{align} \label{eq_analogEncoding}
U_A(x) \ket{0} = \frac{1}{\norm{x}}\sum_{i=0}^{N-1} x_i \ket{i} \, .
\end{align}
Whereas in~\eqref{eq_analogEncoding} the classical data is encoded into the amplitudes of the quantum state (this is usually referred to as analog or amplitude encoding), we can also encode classical data into the basis vectors of a quantum state (this is called digital or binary encoding). Let $d=(d_0,\ldots,d_{N-1})$ denote a digital approximation of $x$, where $d_i$ are $q$-bit strings. Then, we define the digital encoding as the unitary transformation $U_D$, given by
\begin{align*}
U_D \ket{i} \ket{0} = \ket{i} \ket{d_i}\,  \qquad \textnormal{for all}  \quad i \in [N]  \, .
\end{align*}
The core of our algorithm works with digital encodings. We next make a technical simplifying assumption.
\begin{assumption}[Sufficient precision] \label{ass_precision}
We assume to have sufficient precision such that all basic quantum arithmetic operations can be executed without any errors. Furthermore, the number of bits to attain such precision is $O(\polylog(N))$
\end{assumption}
We note that the first part of this assumption is also necessary for the classical discrete LFT algorithm (see Remark~\ref{rmk_precision_classical}), to avoid rounding errors that would make the analysis cumbersome. Since this paper discusses discrete algorithms, this is a natural assumption to avoid complications that would arise with any finite-precision computer, and is not specific to the quantum algorithms. The second part of the assumption is only technical: we want to avoid carrying the dependence on the number of qubits in the running time expressions. We could easily drop the assumption and introduce polynomial dependence on the number of qubits in all the running time expressions. The assumption is still somewhat justified, because if the number of necessary qubits is large (compared to $\log N$) the complexity of quantum arithmetics will no longer be asymptotically negligible, making our algorithms less interesting.

Recall that we consider regularly discretized primal and dual spaces. For such spaces, we can assume without loss of generality that the discretization points $x_i$ are procedurally generated, i.e., in light of Assumption~\ref{ass_discrete}, we have $x_i = x_0 + i \gamma_{\mathrm{x}}$ for a given $x_0$, and $s_j = s_0 + j \gamma_{\mathrm{s}}$ for a given $s_0$. This implies that the operations $\ket{i}\ket{0} \mapsto \ket{i}\ket{x_i}$ and $\ket{j}\ket{0} \mapsto \ket{j}\ket{s_j}$ can be done in $O(\polylog (N))$ for all $i \in [N]$.\footnote{In case the points $x_i$ are only available as an unstructured list and there is no known algorithm to construct them given index $i$, then simply loading the points in primal space would take time $O(N \polylog(N))$, negating any advantage of the quantum algorithm and making the classical $O(N + K)$ a better option.} This observation can be extended to the multidimensional case, where $x_{i_1,\ldots,i_d}=x_{0,\ldots,0}+(i_1,\ldots,i_d)\gamma_{\mathrm{x}}$ for a given  $x_{0,\ldots,0}$ (see Assumption~\ref{ass_discrete_d_dim}), and $s_{j_1,\ldots,j_d}=s_{0,\ldots,0}+(j_1,\ldots,j_d)\gamma_{\mathrm{s}}$ for a given $s_{0,\ldots,0}$.

As for the discrete LFT in the classical case, we do not need to know the function $f$ on the entire domain $[0,1]$, but only at $N$ points. In the quantum case this is achieved by the following assumption.
\begin{assumption}[Access to function] \label{ass_Uf}
We assume having access to a unitary $U_f$ such that 
\begin{align*}
U_f(\ket{x_i}\ket{0})=\ket{x_i}\ket{f(x_i)}\, \qquad \textnormal{for all}  \quad i \in [N]\, .    
\end{align*}
Furthermore, the cost of running $U_f$ is $O(\polylog(N))$.
\end{assumption}
This assumption is standard and considered in multiple quantum algorithms for convex optimization~\cite{childs20,wolf20}.
The second part of the assumption (running time) is well justified because for every classically efficiently computable function, i.e., computable in $O(\polylog(N))$ time, we can use quantum arithmetic to construct $U_f$ with $O(\polylog(N))$ gates~\cite{nielsenChuang_book}. 
The next remark discusses a way to bypass Assumption~\ref{ass_precision}.
\begin{remark}[Circumventing Assumption~\ref{ass_precision}]
As discussed in Remark~\ref{rmk_precision_classical}, Assumption~\ref{ass_precision} is needed in the LFT algorithm because numerical errors could lead to an incorrect optimizer via~\eqref{eq_LFT_optimizer}, which would then be amplified because of the discrete nature of the optimizers. Suppose we have a limited precision such that the discrete gradients $c_{-1},\ldots,c_{N-1}$ all have an error of at most $\eps>0$. If we assume that $\eps$ is sufficiently small that each interval $[ \lfloor (c_i - c_{i-1})/\gamma_{\mathrm{s}}\rfloor-2\eps, \lfloor (c_i - c_{i-1})/\gamma_{\mathrm{s}}\rfloor+2\eps]$ does not contain any integer, then rule~\eqref{eq_LFT_optimizer} chooses the correct optimizer despite the presence of the precision errors.\footnote{Notice that if $(c_i - c_{i-1})/\gamma_{\mathrm{s}}$ is an integer, this assumption never holds. A solution would be to change the primal discretization in such a way that the $c_i$ values change.} Alternatively, we would have to accept errors that scale as $\gamma_{\mathrm{x}}$ in the output of the discrete (quantum or classical) LFT.
\end{remark}


The QLFT, as defined in~\eqref{eq_QLFT}, stores the classical information of the LFT, in form of the values $f^*(s_0),\ldots,f^*(s_{K-1})$, in quantum registers using a digital representation. For certain applications it may be useful to have them represented as amplitudes. This can be achieved by an efficient probabilistic operation called ``quantum digital-analog conversion", described in Remark~\ref{rmk_QDA_conversion}. In technical terms, we can efficiently perform the operation
\begin{align} \label{eq_QLFT_amplitude}
   \frac{1}{\sqrt{K}} \sum_{j=0}^{K-1} \ket{j} \ket{f^*(s_j)}\quad \to \quad \frac{1}{\sqrt{\alpha}} \sum_{j=0}^{K-1} f^*(s_j) \ket{j} \, ,
\end{align}
where $\alpha:=\sum_{j=0}^{K-1} f^*(s_j)^2 $ is a normalization constant.

\begin{remark}[Quantum digital-analog conversion]\label{rmk_QDA_conversion}
It can be useful to transform digitally encoded data into the analog representation. This is achieved by a so-called \emph{quantum digital-analog conversion}, which is used in existing quantum algorithms such as in~\cite{HHL09}.\footnote{We refer the interested reader to~\cite{kosuke19} for more information about digital-analog conversions. The rough idea of the algorithm is rotate the function value to an ancilla system and performing a measurement which is repeated until we see a specific outcome. The process can be sped up by using amplitude amplification.} There is a probabilistic quantum algorithm for the transformation~\eqref{eq_QLFT_amplitude} with an expected running time $O(\sqrt{1/\omega_f} \polylog(K))$, where 
the probability of success of the algorithm is given by
\begin{align*}
\omega_f
    :=\frac{1}{K} \sum_{j=0}^{K-1} \left(\frac{f^*(s_j)}{\max_{\ell \in [K]}|f^*(s_\ell)|} \right)^2 \, .
\end{align*}
Using the equivalence of norms we see that a worst case bound is given by $\omega_f \geq 1/K$, which holds with equality in case the vector $(f^*(s_0),\ldots,f^*(s_{K-1}))$ is zero everywhere except at one entry. This bound will destroy the exponential speedup in $K$. Fortunately, for better behaved functions, whose entries are more uniformly distributed, $\omega_f$ can be independent of $K$; this is the case, for example, for quadratic functions.\footnote{It has been observed in different quantum algorithms that data needs to be sufficiently uniformly distributed in order to obtain a quantum speedup~\cite{aaronson15}.} For the function discussed in Example~\ref{ex_quadratic}, for $s \in [-1/2,1]$ we see that in the continuous case (where $K\to \infty$) we have $\omega_f= \frac{2}{3} \norm{f^{\ast}}_2^2/\norm{f^{\ast}}_{\infty}^2=1841/9920$.\footnote{For a multidimensional example consider the a function $f$ such that its (continuous) LFT is given by $f^{\ast}:[-L,L]^d \to \R$ with $s\mapsto s^{\mathrm{T}} A s$ for $A=\diag(\alpha_0,\ldots,\alpha_{d-1})$ where $\alpha_i \geq 0$ for all $i \in [d]$. A straightforward calculation reveals that \smash{$\omega_f= (2L)^{-d} \norm{f^{\ast}}_2^2/\norm{f^{\ast}}_{\infty}^2 =  (2L)^{-d} L^3 4/3$}, which decays exponentially in $d$. We see that this is inline with the running time for the $d$-dimensional QLFT (see Corollary~\ref{cor_d_dim_QLFT_regular}).} The interested reader can find more information about the quantum digital-analog conversion in~\cite[Section~4]{gilyen19}.
\end{remark}

\subsection{Regular discretization} 
In this section we assume that the sets $\cX^1_N$ and $\cS_K^1$ are regular, i.e., they fulfill Assumption~\ref{ass_discrete}.
Algorithm~\ref{algo_QLFT} below presents an overview of the steps required to compute the QLFT. The details can be found in the proof of Theorem~\ref{thm_QLFT}, which also proves the correctness and the worst-case running time of the algorithm.
\begin{algorithm}[!htb]
\caption{One-dimensional QLFT with regular discretization}
\label{algo_QLFT}
\begin{algorithmic}
\STATE \textbf{Input:} $n, k \in \N$, $N=2^{n}$, $K=2^{k}$, $\cX^1_N=\{x_0,\ldots,x_{N-1}\}$ satisfying Assumption~\ref{ass_discrete}, and a function $f$ satisfying Assumptions~\ref{ass_regularity} and~\ref{ass_Uf};\vspace{2mm} \\
\item \textbf{Output:} State $\frac{1}{\sqrt{K}} \sum_{j=0}^{K-1} \ket{j} \ket{f^*(s_j)}\ket{\text{Garbage}(j)}$ where $\cS^1_K=\{s_0,\ldots,s_{K-1}\}$ satisfies Assumption~\ref{ass_discrete} with $s_0=c_0$ and $s_{K-1}=c_{N-2}$ as defined in~\eqref{eq_ci}; \vspace{2mm}\\
Compute $c_0$ and $c_{N-2}$ as defined in~\eqref{eq_ci}, and let $(s_0,\ldots,s_{K-1})=(c_0,c_0+\gamma_{\mathrm{s}},\ldots,c_{N-2}-\gamma_{\mathrm{s}},c_{N-2})$ for constant $\gamma_{\mathrm{s}}=(c_{N-2}-c_0)/2$.
Afterwards, do the following:
\begin{enumerate}
\item Compute $\frac{1}{\sqrt{N}}\sum_{i=0}^{N-1} \ket{i} \ket{x_{i-1},x_i,x_{i+1}} \ket{f(x_{i-1}),f(x_i),f(x_{i+1})} $;
\item\label{it_step2_QLFT} Compute $\frac{1}{\sqrt{N}}\sum_{i=0}^{N-1} \ket{i} \ket{x_{i-1},x_i,x_{i+1}} \ket{f(x_{i-1}),f(x_i),f(x_{i+1})} \ket{c_{i-1},c_i} $;
\item \label{it_step3_prob} Compute $\frac{1}{\sqrt{K}} \sum_{j=0}^{K-1} \ket{j} \ket{x^\star_j} \ket{f(x^\star_j)}\ket{\text{Garbage}(j)} $, where $x_j^\star$ is defined in~\eqref{eq_LFT_optimizer};$^\ddagger$
\item Compute $\ket{v}= \frac{1}{\sqrt{K}} \sum_{j=0}^{K-1}  \ket{j} \ket{f^*(s_j)} \ket{\text{Garbage}(j)}$;
\end{enumerate}
Output $\ket{v}$;\vspace{1mm}\\
{\small $^\ddagger$As explained in the proof of Theorem~\ref{thm_QLFT}, Step~\ref{it_step3_prob} is probabilistic and may be repeated until it succeeds.}
\end{algorithmic}
\end{algorithm}
\begin{theorem}[Performance of Algorithm~\ref{algo_QLFT}] \label{thm_QLFT}
Let $n, k \in \N$, $N=2^{n}$, $K=2^{k}$, $\cX^1_N=\{x_0,\ldots,x_{N-1}\}$ satisfying Assumption~\ref{ass_discrete} and $f$ satisfying Assumptions~\ref{ass_regularity} and~\ref{ass_Uf}.
Algorithm~\ref{algo_QLFT} is successful with probability $1/\kappa$, where $\kappa$ denotes the condition number of $f$.
Given Assumption~\ref{ass_precision}, the output $\ket{v}$ of a successful run of Algorithm~\ref{algo_QLFT} is equal to
\begin{align*}
 \frac{1}{\sqrt{K}} \sum_{j=0}^{K-1} \ket{j} \ket{f^*(s_j)}\ket{\textrm{Garbage}(j)}\, ,
\end{align*}
where $(s_0,\ldots,s_{K-1})=(c_0,c_0+\gamma_{\mathrm{s}},\ldots,c_{N-1}-\gamma_{\mathrm{s}},c_{N-2})$, for $c_0$ and $c_{N-1}$ as defined in~\eqref{eq_ci}, $\gamma_{\mathrm{s}}\geq 0$, and $\text{Garbage}(j)$ denotes the content of a working register that depends on $j$, defined more precisely in the proof. Combined with amplitude amplification the expected running time of the algorithm is 
\begin{align*}
    O\big(\sqrt{\kappa}\, \polylog(N,K)\big) \, .
\end{align*}
\end{theorem}
In case the function $f$ does not have a condition number the probability of success is given by
\begin{align*}
    \frac{K}{N W}
    \geq \frac{K}{N}\frac{1}{\left \lfloor \bar L_{\gamma_{\mathrm{x}}} \gamma_{\mathrm{x}}/\gamma_{\mathrm{s}} \right \rfloor}
    \geq \frac{\nu}{L'} 
    \geq \frac{1}{\kappa} \, ,
\end{align*}
where $W$ is defined in~\eqref{eq_W}, $\bar L_{\gamma_{\mathrm{x}}}$ and $\nu$ are defined in Assumption~\ref{ass_regularity}, $\gamma_{\mathrm{x}},\gamma_{\mathrm{s}}$ are defined in Assumption~\ref{ass_discrete}. The penultimate inequality holds only if $f$ is differentiable, its derivative is $L'$-Lipschitz continuous and the final inequality is valid if $f$ is strongly convex. The expected running time of the algorithm is 
\begin{align*}
    O\big(\sqrt{N W/K}\,\polylog(N,K)\big)= O\big(\sqrt{\kappa}\, \polylog(N,K)\big) \, .
\end{align*}

\begin{proof}
The correctness of Algorithm~\ref{algo_QLFT} follows straightforwardly from the correctness of the classical algorithm~\cite{lucet97}. The nontrivial part is that all the steps can be done efficiently. Because of Assumption~\ref{ass_precision}, quantum arithmetic operations can be carried out without introducing any errors.
The initialization step can be done in time $O(\polylog(N))$, as $c_0$ and $c_{N-2}$ can be computed via $U_f$. 
We next analyze each step separately:
\begin{enumerate}
\item \label{it_step1_QLFT} This step has complexity $O(\polylog(N))$. Due to the regular known discretization of the primal space we can load the classical vector $x$ efficiently, i.e., we can prepare $\frac{1}{\sqrt{N}}\sum_{i=0}^{N-1}  \ket{i}\ket{x_i}$ in $O(\polylog(N))$. Next, apply the unitary $U_f$ to obtain 
\begin{align*}
\frac{1}{\sqrt{N}}\sum_{i=0}^{N-1}  \ket{i} U_f(\ket{x_i}\ket{0}) = \frac{1}{\sqrt{N}}\sum_{i=0}^{N-1} \ket{i} \ket{x_i} \ket{f(x_i)}\, .
 \end{align*}
Because we can do the mapping $\ket{0}\ket{i} \ket{0} \mapsto \ket{i-1} \ket{i}\ket{i+1}$, we can repeat the construction above to create
\begin{align*}
\frac{1}{\sqrt{N}}\sum_{i=0}^{N-1} \ket{i} \ket{x_{i-1},x_i,x_{i+1}} \ket{f(x_{i-1}),f(x_i),f(x_{i+1})}\, ,
\end{align*}
where $x_{-1},x_N$, $f(x_{-1})$, and $f(x_{N})$ are irrelevant for the calculation and can be set to an arbitrary value.

\item \label{step_2} This step can be done in time $O(\polylog(N))$ by evaluating the function $\frac{f(x_{i+1}) - f(x_{i})}{\gamma_{\mathrm{x}}}$ to compute $c_{i-1}$ and $c_i$. We note that the value of $c_{-1}$ is irrelevant, and can be filled with an arbitrary number.

\item \label{step_3} The complexity for this step is $O(\polylog(N,K))$.
As explained in Section~\ref{sec_DLFT}, at most $W$ different discrete points in $s_j$ can have the same optimizer $x^\star_j$, where $W$ is defined in~\eqref{eq_W}. This is correct under the assumption that $s_j \in [c_0,c_{N-2}]$, which is the range where the LFT is nontrivial.
From the state prepared in Step~\ref{step_2} we can create
\begin{align*}
&\frac{1}{\sqrt{N W}}\!\sum_{i=0}^{N-1}\! \sum_{m=0}^{W-1}\!\! \ket{i} \ket{x_{i-1},x_i,x_{i+1}} \ket{f(x_{i-1}),f(x_i),f(x_{i+1})} \ket{c_{i-1},c_i} \ket{m}\ket{\mathds{1}\{(i,m)\!\! \in\!\! \cA \}} \ket{j(i,m,c_{i-1})} \, ,
\end{align*}
where $\cA$ and $j(i,m,c_{i-1})$ are defined in~\eqref{eq_setA} and~\eqref{eq_def_j}, respectively. Note that this computation can be done efficiently with quantum arithmetics, since all operations involve simple arithmetic operations or Boolean logic (the minimization in the definition of \eqref{eq_def_j} can be efficiently carried out by binary search with cost $O(\polylog(N))$, since all $s_\ell$ can be generated efficiently given index $\ell$, as discussed in Section~\ref{sec_input_output}).
We next uncompute the registers $\ket{x_{i-1},x_{i+1}}  \ket{f(x_{i-1}),f(x_{i+1})} \ket{c_{i-1},c_i}$, which gives
\begin{align} \label{eq_middle}
\frac{1}{\sqrt{N W}}\sum_{i=0}^{N-1} \sum_{m=0}^{W-1} \ket{i} \ket{x_i} \ket{f(x_i)} \ket{m}  \ket{\mathds{1}\{(i,m) \in \cA \}} \ket{j(i,m,c_{i-1})} \, .
\end{align}
If we perform a measurement on the register with the indicator function, then, conditioned on seeing the outcome ``$1$" and after a relabelling of the sum, we obtain
\begin{align}
\frac{1}{\sqrt{K}}\sum_{j=0}^{K-1} \ket{j} \ket{x^\star_j} \ket{f(x^\star_j)} \ket{m(j)}\ket{i(j)}  \, , \label{eq_endStep3}
\end{align}
where $x^\star_j$ denotes the optimizer given in~\eqref{eq_LFT_optimizer}. (Recall that $|\cA| = K$.)
This step is probabilistic, and succeeds with probability 
\begin{align*}
     \frac{K}{N W} 
     \geq \frac{K}{N}\frac{1}{\left \lfloor \bar L_{\gamma_{\mathrm{x}}} \gamma_{\mathrm{x}}/\gamma_{\mathrm{s}} \right \rfloor} \, ,
\end{align*}
where we used the fact that the indicator function in~\eqref{eq_middle} maps $N\times W$ nonzero indices to $K$ nonzero indices. The inequality then follows from~\eqref{eq_W}.
If $f$ is differentiable and its gradient is $L'$ Lipschitz continuous, we can further simplify this probability as
\begin{align*}
  \frac{K}{N W} 
  \geq \frac{s_{K-1}-s_0}{L'(x_{N-1}-x_0)}  
  = \frac{c_{N-2}-c_0}{L'(x_{N-1}-x_0)}
  =\frac{\nu}{L'}
  \geq \frac{\mu}{L'}
  = \frac{1}{\kappa}\, ,
\end{align*}
where the first inequality uses~\eqref{eq_W_bound}, and the second inequality follows from~\eqref{eq_nu} and holds only if $f$ is strongly convex.


\item This step can be done in $O(\polylog(K))$. To see this, note we can compute $s x^\star-f(x^\star)$ from the state~\eqref{eq_endStep3} to obtain
\begin{align*}
\frac{1}{\sqrt{K}}\sum_{j=0}^{K-1} \ket{j} \ket{s_j x^\star_j - f(x^\star_j)}  \ket{x^\star_j} \ket{f(x^\star_j)} \ket{m(j)}\ket{i(j)} 
&=\frac{1}{\sqrt{K}}\sum_{j=0}^{K-1} \ket{j} \ket{f^*(s_j)}  \ket{x^\star_j} \ket{f(x^\star_j)} \ket{m(j)}\ket{i(j)}  \, .
\end{align*}
Because we are given $\ket{x^\star_j}$, we can uncompute $\ket{f(x^\star_j)}$ and find
\begin{align}
\frac{1}{\sqrt{K}}\sum_{j=0}^{K-1} \ket{j} \ket{f^*(s_j)}  \ket{x^\star_j} \ket{m(j)}\ket{i(j)}  \, . \label{eq_done}
\end{align}
where $ \ket{x^\star_j}\ket{m(j)}\ket{i(j)}$ may be viewed as garbage that depends on $j$.
\end{enumerate}
We can use amplitude amplification~\cite{brassard02} in Step~\ref{step_3} so that $O(\sqrt{NW/K})= O(\sqrt{\kappa})$ repetitions of Algorithm~\ref{algo_QLFT} are sufficient to succeed with constant probability. As a result the overall expected running time is given by $O(\sqrt{NW/K}\,\polylog(N,K))= O(\sqrt{\kappa}\,\polylog(N,K))$.
\end{proof}

\begin{remark}[Algorithm~\ref{algo_QLFT} can succeed with probability $1$]
As stated in Theorem~\ref{thm_QLFT} the probability of success of Algorithm~\ref{algo_QLFT} depends on intrinsic properties of the function $f$. For well-behaved functions this probability can be $1$ which is of particular relevance in a multidimensional scenario --- as discussed in Section~\ref{sec_QLFT_multi}. Examples of such functions are
\begin{enumerate}[(i)]
    \item quadratic functions of the form $f(x)= ax^2+bx+c$ (they have condition number $\kappa=1$, since the second derivative is constant);
    \item certain piecewise linear functions, e.g., the one discussed in Example~\ref{ex_piecewise_linear_1}. We note that there are also piecewise linear functions where the probability of success of Algorithm~\ref{algo_QLFT} is strictly less than $1$ (see Example~\ref{ex_piecewise_linear_2}). 
\end{enumerate}
In case Algorithm~\ref{algo_QLFT} does not succeed with probability $1$, one can always switch to an adaptive dual set which enforces the algorithm to be deterministic. This is discussed in detail in Section~\ref{sec_QLFT_adaptive}.
\end{remark}

As discussed in Remark~\ref{rmk_onedim_classical}, there is an efficient classical algorithm to compute $f^*(s_j)$ at a single $s_j$. If we have a regularly discretized dual set, we can apply a quantum implementation of the efficient classical algorithm to the superposition $\frac{1}{\sqrt{K}} \sum_{j=0}^{K-1} \ket{j}\ket{0}$ to compute $\frac{1}{\sqrt{K}} \sum_{j=0}^{K-1} \ket{j}\ket{f^*(s_j)}$. The running time for this trivial algorithm is $O(\polylog(N, K))$, as we can efficiently compute the discrete gradients and perform binary search to solve \eqref{eq_LFT_optimizer} with a small number of qubits as working space. The QLFT algorithm has the same asymptotic running time (up to polylogarithmic factors) only when $\kappa = 1$, therefore it would seem that the trivial quantization of a classical algorithm is a better choice than our QLFT; however, this only holds when $d=1$. Indeed, as will be shown in Section~\ref{sec_QLFT_multi}, the running time of the QLFT proposed in this paper can scale much better in $d$ than \emph{any} classical algorithm, see also the discussion in Section~\ref{sec_optimality}. This emphasizes the fact that the advantage for our quantum algorithm holds in the multidimensional setting: the problem is already classically easy when $d=1$, and there is no meaningful quantum speedup in this case.

\begin{remark}[QLFT for nonconvex functions] \label{rmk_quatum_nonConvex}
Following Remark~\ref{rmk_nonConvex}, it is natural to ask if the QLFT can be computed efficiently for nonconvex functions. Classically, this can be achieved by running Algorithm~\ref{algo_DLFT} on the convex hull of the function. There is a quantum algorithm that computes the convex hull of $N$ points in $\R^2$ in $O(\sqrt{N}h)$, where $h$ is the number of points comprising the hull~\cite{lanzagorta_uhlmann_2010}; this is obtained using Grover search to achieve a quadratic speedup in $N$ compared to a classical algorithm. Furthermore, a heuristic version of the quantum algorithm, that often works in practice, runs in time $O(\sqrt{Nh})$~\cite{lanzagorta_uhlmann_2010}. Since the QLFT algorithm works on a uniform superposition of points, one may ask if there is hope for an ``implicit'' quantum convex hull algorithm that does not output a classical description of the convex hull; rather, an algorithm that takes as input a quantum state encoding of a superposition of points, and outputs a {\em uniform} superposition of only the points in the convex hull. However, it is easily established that this cannot be done in polylogarithmic time. Indeed, suppose we are given oracle access to a function $g: [ N ] \to \{0,1\}$ that takes value $1$ only at the point $z$. By preparing the superposition $\frac{1}{\sqrt{N}} \sum_{i=0}^{N-1} \ket{i}\ket{g(i)}$, applying  the ``implicit'' quantum convex hull on the second register, then measuring the second register, we would observe $\ket{z}$ in the first register with probability $1/2$ because the convex hull of the second register is $[0,1]$, and $z$ is the only one point whose second-register coordinate is $1$. It is well known that $\Omega(\sqrt{N})$ queries to $g$ are necessary to recover $z$ \cite{zalka1999grover}, hence an ``implicit'' convex hull algorithm that runs in time $O(\polylog(N))$ is ruled out. 
\end{remark}

\subsection{Adaptive discretization} \label{sec_QLFT_adaptive}
As discussed in Section~\ref{sec_DLFT_adaptive}, for the classical algorithm we can choose the dual set $\cS^1_{K}$ in an adaptive way, so that we have a unique optimizer in~\eqref{eq_DLFT}. This has the advantage that the quantum algorithm succeeds with probability $1$, can be described by a unitary evolution, and does not accumulate any garbage. The downside is that we cannot control the discretization resolution in the dual space.

Algorithm~\ref{algo_QLFT_adaptive} below gives an overview of the steps required to compute the QLFT for an adaptively chosen dual space. The details can be found in the proof of Theorem~\ref{thm_QLFT}, which also proves the correctness and the worst-case running time of the algorithm.
\begin{algorithm}[!htb]
\caption{One-dimensional QLFT with adaptive discretization}
\label{algo_QLFT_adaptive}
\begin{algorithmic}
\STATE \textbf{Input:} $n\in \N$, $N=2^{n}$, $\cX^1_N=\{x_0,\ldots,x_{N-1}\}$ satisfying Assumption~\ref{ass_discrete}, and a function $f$ satisfying Assumptions~\ref{ass_regularity} and~\ref{ass_Uf};\vspace{2mm} \\
\item \textbf{Output:} State $\frac{1}{\sqrt{N}} \sum_{i=0}^{N-1} \ket{i} \ket{x_i} \ket{s_i} \ket{f^*(s_i)}$ for $\cS_{N,\mathrm{adaptive}}^1=\{s_0,\ldots,s_{N-1}\}$ as defined in~\eqref{eq_adaptive_dual_set}; \vspace{2mm}\\
Do the following:
\begin{enumerate}
\item Compute $\frac{1}{\sqrt{N}}\sum_{i=0}^{N-1} \ket{i} \ket{x_{i-1},x_i,x_{i+1}} \ket{f(x_{i-1}),f(x_i),f(x_{i+1})} $;
\item Compute $\frac{1}{\sqrt{N}}\sum_{i=0}^{N-1} \ket{i} \ket{x_i} \ket{f(x_i)} \ket{s_i}$;
\item Compute $\ket{v}= \frac{1}{\sqrt{N}} \sum_{i=0}^{N-1}   \ket{i} \ket{x_i} \ket{s_i} \ket{f^*(s_i)}$;
\end{enumerate}
Output $\ket{v}$;
\end{algorithmic}
\end{algorithm}

\begin{remark}
If desired, we can uncompute the $\ket{x_i}$ register so that the output of Algorithm~\ref{algo_QLFT_adaptive} is an approximation to
\begin{align*}
    \frac{1}{\sqrt{N}} \sum_{i=0}^{N-1}   \ket{i} \ket{s_i} \ket{f^*(s_i)}\, .
\end{align*} 
\end{remark}

\begin{theorem}[Performance of Algorithm~\ref{algo_QLFT_adaptive}] \label{thm_QLFT_adaptive}
Let $n \in \N$, $N=2^{n}$, $\cX^1_N=\{x_0,\ldots,x_{N-1}\}$ satisfying Assumption~\ref{ass_discrete}, and $f$ satisfying Assumptions~\ref{ass_regularity} and~\ref{ass_Uf}.
Then Algorithm~\ref{algo_QLFT_adaptive} is successful with probability $1$ and, given Assumption~\ref{ass_precision}, its output $\ket{v}$ is equal to
\begin{align*}
\frac{1}{\sqrt{N}} \sum_{i=0}^{N-1}   \ket{i} \ket{x_i} \ket{s_i} \ket{f^*(s_i)} \, ,
\end{align*}
where $(s_0,\ldots,s_{N-1})=(\frac{c_{-1}+c_0}{2},\frac{c_0+c_1}{2},\ldots,\frac{c_{N-3}+c_{N-2}}{2},\frac{c_{N-2}+c_{N-1}}{2})$ with $c_i$ defined in~\eqref{eq_ci}.
Furthermore, Algorithm~\ref{algo_QLFT_adaptive} runs in time $O(\polylog(N))$.
\end{theorem}
\begin{proof}
Recalling that by Assumption~\ref{ass_precision}, quantum arithmetic operations can be carried out without introducing any errors.
The state
\begin{align} \label{eq_state1}
\frac{1}{\sqrt{N}}\sum_{i=0}^{N-1} \ket{i} \ket{x_{i-1},x_i,x_{i+1}} \ket{f(x_{i-1}),f(x_i),f(x_{i+1})}
\end{align}
can be constructed in time $O(\polylog(N))$, as discussed in Step~\ref{it_step1_QLFT} of the proof of Theorem~\ref{thm_QLFT}. Since for the (centered) adaptive discretization we have $s_i=(c_{i-1}+c_i)/2$ for $i\in[N]$, we can first construct the state given in Step~\ref{it_step2_QLFT} of Algorithm~\ref{algo_QLFT}, and from there compute
\begin{align}\label{eq_state2}
\frac{1}{\sqrt{N}}\sum_{i=0}^{N-1} \ket{i} \ket{x_i} \ket{f(x_i)} \ket{s_i}\, ,
\end{align}
where we uncomputed all the registers that are not needed anymore.\footnote{We are flexible to also consider slight variants of the (centered) adaptive discretization as discussed in Remark~\ref{rmk_adaptive_disc}.} At this point we can use the fact that for the adaptive protocol the dual discretization is chosen such that $x^\star_i = x_i$, yielding $f^*(s_i)=s_i x_i -f(x_i)$. This allows us to create
\begin{align}\label{eq_state3}
\frac{1}{\sqrt{N}}\sum_{i=0}^{N-1} \ket{i} \ket{x_i}  \ket{s_i} \ket{f^*(s_i)}\, ,
\end{align}
where we again uncompute registers that are no longer needed.

\end{proof}

\section{Multidimensional quantum Legendre-Fenchel transform} \label{sec_QLFT_multi}
For $d\in \N$, consider a multivariate function $f:[0,1]^d \to \R$ that satisfies Assumption~\ref{ass_regularity}. The factorization property of the LFT ensures that for $s \in \R^d$ we have
\begin{align*}
f^{\ast}(s)
&=\sup_{x \in [0,1]^d}\{\langle s, x \rangle - f(x) \} \\
&= \sup_{x_0 \in [0,1]}\Big\{ s_0 x_0 + \sup_{x_1 \in [0,1]}\big\{ s_1 x_1 + \ldots + \sup_{x_{d-1} \in [0,1]}\{ s_{d-1} x_{d-1} -f(x)\} \ldots \big\} \Big\}\, .
\end{align*}
From this we see that we can use $d$-times the one-dimensional discrete LFT algorithm to efficiently compute the discrete LFT in a $d$-dimensional setting.

To improve the readability of this manuscript we sometimes restrict ourselves to $d=2$, however, all the statements can be extended to an arbitrary dimension $d\in \N$. 
Furthermore in this section we will use a function $g$ defined as the negative LFT of $f(x_0,x_1)$ in the variable $x_1$ for fixed $x_0$, i.e.,
\begin{align} \label{eq_def_g}
    g(x_0,s_1):= - \sup_{x_1 \in [0,1]}\{ s_1 x_1 -f(x_0,x_1) \}\, .
\end{align}
With the help of $g$ we see that to compute the LFT of $f^{\ast}$ it suffices to first compute the one-dimensional LFT of $f(x_0,x_1)$ in the variable $x_1$, while keeping $x_0$ fixed which introduces the function $g$. Afterwards we compute the LFT of the function $g$ in the variable $x_0$ while keeping $s_1$ fixed. In technical terms this can be expressed as
\begin{align*} 
f^{\ast}(s_0,s_1)
= \sup_{x_0 \in [0,1]}\{ s_0 x_0 + \sup_{x_1 \in [0,1]}\{ s_1 x_1 -f(x_0,x_1)\} \}
=: \sup_{x_0 \in [0,1]}\{ s_0 x_0 - g(x_0,s_1) \}\, .
\end{align*}
For any constant $s_1$ the function $x_0 \mapsto g(x_0,s_1)$ is convex, because $f$ is by assumption jointly convex~\cite[Proposition~5]{tropp12}. Furthermore, $f^{\ast}(s_0,s_1)$ is the LFT of $g(x_0,s_1)$ in the variable $x_0$.

\subsection{Regular discretization} \label{sec_multi_dim_QLFT_regular}
As discussed in the one-dimensional case in Section~\ref{sec_DLFT_regular}, we start by defining the regular discrete primal and dual sets.
Let $N=\prod_{\ell=0}^{d-1} N_{\ell}$ and $K=\prod_{\ell=0}^{d-1} K_\ell$, where $N_{\ell}$ and $K_{\ell}$ denote the number of grid points per dimension.
\begin{assumption}[Two-dimensional regular discretization] \label{ass_discrete_d_dim}
The sets $\cX^2_N=\{x_{i_0,i_{1}}\}_{i_{0} \in [N_0],i_{1} \in [N_1] }$ and $\cS^2_K=\{s_{j_0,j_{1}}\}_{j_{0} \in [K_{0}],j_1 \in [K_1]}$ are such that:
\begin{enumerate}[(i)]
\item the discretization is \emph{sorted}, i.e., 
\begin{itemize}
    \item $x_{i_0,i_1} \leq x_{i_0+1,i_1}$ $\forall$ $i_{0} \in [N_{0}-1], i_1 \in [N_1]$ and $x_{i_0,i_1} \leq x_{i_0,i_1+1}$ $\forall$ $i_{0} \in [N_{0}], i_1 \in [N_1-1]$
    \item $s_{j_0,j_1} \leq s_{j_0+1,j_1}$ $\forall$ $j_{0} \in [K_{0}-1], j_1 \in [K_1]$ and $s_{j_0,j_1} \leq s_{j_0,j_1+1}$ $\forall$ $j_{0} \in [K_{0}], j_1 \in [K_1-1]$;
\end{itemize}
\item the discretization is \emph{regular}, i.e., 
\begin{itemize}
    \item $x_{i_0+1,i_{1}} - x_{i_0,i_{1}} = \gamma_{\mathrm{x}}$ for all $i_{0} \in [N_{0}]$ and $x_{i_0,i_{1}+1} - x_{i_0,i_{1}} = \gamma_{\mathrm{x}}$ for all $i_{1} \in [N_{1}]$
     \item $s_{j_0+1,j_{1}} - s_{j_0,j_{1}} = \gamma_{\mathrm{s}}$ for all $j_{0} \in [K_{0}]$ and $s_{j_0,j_{1}+1} - s_{j_0,j_{1}} = \gamma_{\mathrm{s}}$ for all $j_{1} \in [K_{1}]$.\footnote{The regularity assumption can be slightly relaxed so that $\gamma_{\mathrm{x}}$ and $\gamma_{\mathrm{s}}$ are different in each dimension. For the sake of readability, we work under the assumption of uniform discretization steps.}
\end{itemize}
\end{enumerate}
\end{assumption}

Algorithm~\ref{algo_QLFT_2d} gives an overview of how to compute a multidimensional QLFT; the details are given in the proof of Corollary~\ref{cor_d_dim_QLFT_regular}, which also discusses correctness and running time. The $d$-dimensional QLFT is reduced to the task of performing $d$-times a one-dimensional QLFT, which we know how to do using Algorithm~\ref{algo_QLFT}.

\begin{algorithm}[!htb]
\caption{Two-dimensional QLFT with regular discretization}
\label{algo_QLFT_2d}
\begin{algorithmic}
\STATE \textbf{Input:} $n_0,n_1,k_0,k_1 \in \N$, $N_0=2^{n_0}, N_1=2^{n_1}, K_0=2^{k_0},K_1=2^{k_1}$, $N=N_0 N_1$, set $\cX^2_{N}$ satisfying Assumption~\ref{ass_discrete_d_dim}, and $f$ satisfying Assumptions~\ref{ass_regularity},~\ref{ass_regularity_fac}, and~\ref{ass_Uf};\vspace{2mm} \\
\item \textbf{Output:} State $\frac{1}{\sqrt{K_0 K_1}} \sum_{j_0=0}^{K_0-1} \sum_{j_1=0}^{K_1-1} \ket{j_0,j_1} \ket{f^*(s_{j_0},s_{j_1})}\ket{s_{j_0},s_{j_1}}  \ket{\text{Garbage}(j_0,j_1)}$ where $(s_{j_0},s_{j_1})$ are regular grid points spanning the nontrivial domain of $f^*$; \vspace{2mm}\\
Do the following:
\begin{enumerate}
\item \label{it_step1_dQLFTreg}Compute $\frac{1}{\sqrt{N_0 N_1}}\sum_{i_0=0}^{N_0-1}\sum_{i_1=0}^{N_1-1} \ket{i_0,i_1} \ket{x_{i_0-1,i_1},x_{i_0,i_1-1},x_{i_0,i_1},x_{i_0,i_1+1},x_{i_0+1,i_1}}$\\ \hspace{50mm}$\ket{f(x_{i_0-1,i_1}),f(x_{i_0,i_1-1}),f(x_{i_0,i_1}),f(x_{i_0,i_1+1}),f(x_{i_0+1,i_1})} $;
\item \label{it_step2_dQLFTreg}Perform $1$-dimensional QLFT (see Algorithm~\ref{algo_QLFT}) to obtain\\ $\frac{1}{\sqrt{N_0 K_1 }}\!\sum_{j_1=0}^{K_1-1}\! \sum_{i_0=0}^{N_0-1}\! \ket{j_1,i_0} \ket{g(x_{i_0-1},s_{j_1}),g(x_{i_0},s_{j_1}),g(x_{i_0+1},s_{j_1})} \ket{s_{j_1}(x_{i_0})}\ket{\text{Garbage}(i_0,j_1)} $, where $g(\cdot)$ is defined in~\eqref{eq_def_g};
\item Perform $1$-dimensional QLFT (see Algorithm~\ref{algo_QLFT}) to obtain\\ $\ket{v}=\frac{1}{\sqrt{K_0 K_1}}\sum_{j_0=0}^{K_0-1}\sum_{j_1=0}^{K_1-1}  \ket{j_0,j_1} \ket{f^*(s_{j_0},s_{j_1})} \ket{s_{j_0},s_{j_1}}\ket{\text{Garbage}(j_0,j_1)}$;
\end{enumerate}
Output $\ket{v}$; 
\end{algorithmic}
\end{algorithm}

\begin{corollary}[Performance of Algorithm~\ref{algo_QLFT_2d}] \label{cor_d_dim_QLFT_regular}
Let $d,n_{\ell},k_{\ell} \in \N$, $N_{\ell}=2^{n_{\ell}}$, $K_{\ell}=2^{k_{\ell}}$, for all $\ell \in [d]$, $N=\prod_{\ell=0}^{d-1} N_{\ell}$, $K=\prod_{\ell=0}^{d-1} K_{\ell}$, $\cX^d_N$ satisfying Assumption~\ref{ass_discrete_d_dim}, $f$ satisfying Assumptions~\ref{ass_regularity},~\ref{ass_regularity_fac}, and~\ref{ass_Uf}.
Then Algorithm~\ref{algo_QLFT_2d} applied to the $d$-dimensional case succeeds with probability $1/\kappa^{d}$, where $\kappa$ is the condition number of $f$. Given Assumption~\ref{ass_precision}, its output $\ket{v}$ of a successful run is equal to 
\begin{align*}
 \frac{1}{\sqrt{K}} \sum_{j_0=0}^{K_0-1}\ldots \sum_{j_{d-1}=0}^{K_{d-1}-1}\ket{j_0,\ldots,j_{d-1}} \ket{f^*(s_{j_0},\ldots,s_{j_{d-1}})} \ket{s_{j_0},\ldots,s_{j_{d-1}}}\ket{\textnormal{Garbage}(j_0,\ldots,j_{d-1})} , 
\end{align*}
where $\!(s_{j_0},\ldots,s_{j_{d-1}})\!$ are regular grid points spanning the nontrivial domain of $f^*$. $\text{Garbage}(j_0,\ldots,j_{d-1})$ denotes the content of a working register that depends on $(j_0,\ldots,j_{d-1})$, defined more precisely in the proof.
Combined with amplitude amplification the expected running time of the algorithm is $O(\kappa^{d/2} \, \polylog(N,K))$. 
\end{corollary}
\begin{proof}
To simplify the exposition we state the proof for the case $d=2$, and discuss its extension to arbitrary dimension $d\in \N$ (as stated in the corollary) at the end.
\begin{enumerate}
\item Step~\ref{it_step1_dQLFTreg} is equivalent to the first step of Algorithm~\ref{algo_QLFT} which can be done in $O(\polylog(N))$. In $d$-dimensions we need $O(d)$-registers to create the required state.
\item To see how Step~\ref{it_step2_dQLFTreg} works, note that we add two points of redundancy to the $1$-dimensional QLFT, by choosing the $s$-vector for a fixed $i_0$ as $s_{i_0,0}=c_{i_0,0}-\gamma_{\mathrm{s},i_0}$, $s_{i_0,1}=c_{i_0,0}$, $s_{i_0,K_1-2}=c_{i_0,N_1-2}$, and $s_{i_0,K_1-1}=c_{i_0,N_1-2}+\gamma_{\mathrm{s},i_0}$ (see Remark~\ref{rmk_dualSpace}). With this convention we find that
\begin{align*}
    s_{j_1}(x_{i_0}) \in \Big[\frac{f(x_{i_0},\gamma_{\mathrm{x}})-f(x_{i_0},0)}{\gamma_{\mathrm{x}}},\frac{f(x_{i_0},1)-f(x_{i_0},1-\gamma_{\mathrm{x}})}{\gamma_{\mathrm{x}}}\Big]
\end{align*}
which is increasing in $j_1\in[K_1]$ in $K_1$ regular steps. Algorithm~\ref{algo_QLFT} can be used to do the QLFT of $f(x_{i_0},x_{i_1})$ in the $x_{i_1}$ variable for fixed $x_{i_0}$. More precisely, we can create the state
\begin{align*}
\frac{1}{\sqrt{K_1 N_0}}\sum_{j_1=0}^{K_1-1} \sum_{i_0=0}^{N_0-1} \ket{j_1,i_0} \ket{g(x_{i_0-1},s_{j_1}),g(x_{i_0},s_{j_1}),g(x_{i_0+1},s_{j_1})} \ket{s_{j_1}(x_{i_0})}\ket{\text{Garbage}(i_0,j_1)} \, ,
\end{align*}
where $g(\cdot)$ is defined in~\eqref{eq_def_g}.
This step has a probability of success given by $1/\kappa$. 
To see this, it suffices to show that the QLFT with three adjacent points, i.e., the mapping
\begin{align} \label{eq_desiredOp}
    \frac{1}{N} \sum_{i=0}^{N-1} \ket{i} \ket{f(x_{i-1},x_i,x_{i+1})} \ket{\text{Garbage}(i)}
    \to \frac{1}{K} \sum_{j=0}^{K-1} \ket{j} \ket{f^*(s_{j})} \ket{\text{Garbage}(j)}
\end{align}
can be done with probability of success given by $1/\kappa$. This is indeed possible by first creating the state
\begin{align*}
&\frac{1}{\sqrt{NW}} \sum_{i=0}^{N-1} \sum_{m=0}^{W-1} \ket{i} \ket{x_{i-1},x_i,x_{i+1}} \ket{f(x_{i-1},x_i,x_{i+1})} \ket{m} \ket{\mathds{1}\{(h,m) \in \cA \, \forall h \in \{i-1,i,i+1\}\}} \\
&\hspace{100mm}\ket{j(i,m,c_{i-1}} \ket{\text{Garbage}(i)} \, .
\end{align*}
Note that the indicator function can indeed have the value $1$, i.e., the condition $(h,m) \in \cA \, \forall h \in \{i-1,i,i+1\}$ is verified for some value of $i$ and $m$. By convexity of the function $f$, if the condition $\lfloor \frac{c_h-c_{h-1}}{\gamma_{\mathrm{s}}} \rfloor \geq m+1$ in the definition of $\cA$ is verified for $h = i - 1$,
it is verified for all other $h \in \{i,i+1\}$. Since the condition must be verified for some choice of $i$ and $m$ (by definition of the LFT), this implies that the indicator must have value $1$ for some choice of $i$ and $m$.
We next perform a measurement on the register with the indicator function and postselect on seeing the outcome $1$, which happens with probability $1/\kappa$, to obtain
\begin{align*}
    \frac{1}{K} \sum_{j=0}^{K-1} \ket{j} \ket{x^\star_j} \ket{f(x^\star_{j})} \ket{\text{Garbage}(j)} \, ,
\end{align*}
where $x^\star_j$ is the optimizer in the definition of the LFT, i.e. $f^*(s_j)=s_jx^\star_j-f(x^\star_j)$. From this we can compute the right-hand side of~\eqref{eq_desiredOp}.

\item For the final step we note that the dual variable $s_{j_0}$ is monotonically increasing in equidistant steps, so that
\begin{align*}
s_{j_0} \in \Big[\frac{f(\gamma_{\mathrm{x}},0)-f(0,0)}{\gamma_{\mathrm{x}}},\frac{f(1,1)-f(1-\gamma_{\mathrm{x}},1)}{\gamma_{\mathrm{x}}}\Big]\, ,
\end{align*}
for $j_0 \in [K_0]$, where we used the facts that $g(0,s_{j_1})-g(\gamma_{\mathrm{x}},s_{j_1})=f(\gamma_{\mathrm{x}},0)-f(0,0)$ and $g(1-\gamma_{\mathrm{x}},s_{j_1})-g(\gamma_{\mathrm{x}},s_{j_1})=f(1,1)-f(1-\gamma_{\mathrm{x}},1)$. Algorithm~\ref{algo_QLFT} can then be used to perform the QLFT of $g(x_{i_0},s_{j_1})$ in the $x_{i_0}$ variable for a fixed $s_{j_1}$, to obtain
\begin{align*}
    \frac{1}{\sqrt{K_0 K_1}}\sum_{j_0=0}^{K_0-1}\sum_{j_1=0}^{K_1-1}\ket{j_0,j_1} \ket{f^*(s_{j_0},s_{j_1})} \ket{s_{j_0},s_{j_1}}\ket{\text{Garbage}(j_0,j_1)}\, .
\end{align*}
This step is successful with probability $1/\kappa$, decreasing the overall success proability to $1/\kappa^2$.
\end{enumerate}
To extend the proof to arbitrary $d$, we proceed by induction. It is straightforward to note that bullet points 2-3 above show precisely the induction step from dimension $k$ to $k+1$; each additional step in the induction simply requires carrying two additional points in the registers, as stated, for the computation of the discrete gradients. The assumptions of the single-dimensional QLFT algorithm are trivially satisfied at each induction step.\footnote{A formal induction proof would require cumbersome notation for the $d$-dimensional generalization of the ``partial transform'' in \eqref{eq_def_g}.}

We can use amplitude amplification~\cite{brassard02} so that \smash{$O(\kappa^{d/2})$} repetitions (rather than $O(\kappa^d)$) are sufficient to succeed with constant probability. As a result the overall expected running time is given by \smash{$O(\kappa^{d/2}\,\polylog(N,K))$}.
\end{proof}

\begin{remark}[Computing the multidimensional LFT at a single point] \label{rmk_multidim_classical}
The classical algorithm described in Remark~\ref{rmk_onedim_classical} to compute the one-dimensional LFT at a single point does not extend to an efficient algorithm in $d$ dimensions. Indeed, a straightforward application of the same algorithm, decomposing a $d$-dimensional LFT computation into $d$ one-dimensional LFTs, yields a running time of $O((2 \log N)^d)$; see also the discussion in Section~\ref{sec_optimality}.
\end{remark}

\subsection{Adaptive discretization} \label{sec_multi_dim_QLFT_adaptive}
We can perform a multidimensional QLFT with an adaptive dual discretization which is the generalization of the method explained in Section~\ref{sec_QLFT_adaptive} for the one-dimensional case. The benefit of an adaptively chosen dual space is that the algorithm succeeds with probability $1$. Note that in case the condition number of $f$ satisfies $\kappa>1$, Algorithm~\ref{algo_QLFT_2d} needs to be repeated on average $\kappa^{d/2}$ times which scales exponentially in the dimension $d$.

\begin{algorithm}[!htb]
\caption{Two-dimensional QLFT with adaptive discretization}
\label{algo_QLFT_2d_adaptive}
\begin{algorithmic}
\STATE \textbf{Input:} $n_0,n_1 \in \N$, $N_0=2^{n_0}, N_1=2^{n_1}$, $N=N_0 N_1$, set $\cX^2_{N}$ satisfying Assumptions~\ref{ass_discrete_d_dim}, $f$ satisfying Assumptions~\ref{ass_regularity} and~\ref{ass_Uf};\vspace{2mm} \\
\item \textbf{Output:} State $\frac{1}{\sqrt{N_0 N_1}} \sum_{i_0=0}^{N_0-1} \sum_{i_1=0}^{N_1-1} \ket{i_0,i_1} \ket{f^*(s_{i_0},s_{i_1})}\ket{s_{i_0},s_{i_1}}$ where $(s_{i_0},s_{i_1})$ are adaptively chosen grid points spanning the nontrivial domain of $f^*$; \vspace{2mm}\\
Do the following:
\begin{enumerate}
\item Compute  $\frac{1}{\sqrt{N_0 N_1}}\sum_{i_0=0}^{N_0-1}\sum_{i_1=0}^{N_1-1} \ket{i_0,i_1} \ket{x_{i_0-1,i_1},x_{i_0,i_1-1},x_{i_0,i_1},x_{i_0,i_1+1},x_{i_0+1,i_1}}$\\ \hspace{50mm}$\ket{f(x_{i_0-1,i_1}),f(x_{i_0,i_1-1}),f(x_{i_0,i_1}),f(x_{i_0,i_1+1}),f(x_{i_0+1,i_1})} $;
\item Perform $1$-dimensional QLFT (see Algorithm~\ref{algo_QLFT_adaptive}) to obtain\\ $\frac{1}{\sqrt{N_0 N_1 }} \sum_{i_0=0}^{N_0-1}\sum_{i_1=0}^{N_1-1} \ket{i_0,i_1} \ket{g(x_{i_0-1},s_{i_1}),g(x_{i_0},s_{i_1}),g(x_{i_0+1},s_{i_1})} \ket{s_{i_1}(x_{i_0})}$,\\ where $g(\cdot)$ is defined in~\eqref{eq_def_g};
\item Perform $1$-dimensional QLFT (see Algorithm~\ref{algo_QLFT_adaptive}) to obtain\\ $\ket{v}=\frac{1}{\sqrt{N_0 N_1}}\sum_{i_0=0}^{N_0-1}\sum_{i_1=0}^{N_1-1}  \ket{i_0,i_1} \ket{f^*(s_{i_0},s_{i_1})} \ket{s_{i_0},s_{i_1}}$;
\end{enumerate}
Output $\ket{v}$;
\end{algorithmic}
\end{algorithm}

\begin{corollary}[Performance of Algorithm~\ref{algo_QLFT_2d_adaptive}] \label{cor_d_QLFT_adaptive}
Let $d,n_{\ell} \in \N$, $N_{\ell}=2^{n_{\ell}}$, for all $\ell \in [d]$, $N=\prod_{\ell=0}^{d-1} N_{\ell}$, $\cX^d_N$ satisfying Assumption~\ref{ass_discrete_d_dim}, $f$ satisfying Assumptions~\ref{ass_regularity} and~\ref{ass_Uf}.
Then Algorithm~\ref{algo_QLFT_2d_adaptive} applied on a $d$-dimensional function succeeds with probability $1$. Given Assumption~\ref{ass_precision}, its output $\ket{v}$ is equal to
\begin{align*}
 \frac{1}{\sqrt{N}} \sum_{i_0=0}^{N_0-1}\ldots \sum_{i_{d-1}=0}^{N_{d-1}-1}\ket{i_0,\ldots,i_{d-1}} \ket{f^*(s_{i_0},\ldots,s_{i_{d-1}})} \ket{s_{i_0},\ldots,s_{i_{d-1}}}\, ,
\end{align*}
where $(s_{i_0},\ldots,s_{i_{d-1}})$ are adaptively chosen grid points spanning the nontrivial domain of $f^*$.
The time complexity for a successful run is $O(\polylog(N))$. 
\end{corollary}
\begin{proof}
The proof follows by the same lines as the proof of Corollary~\ref{cor_d_dim_QLFT_regular}, with the only difference that we use Algorithm~\ref{algo_QLFT_adaptive} instead of Algorithm~\ref{algo_QLFT} as a building block.
\end{proof}

\begin{remark}
Suppose we want to compute $f^*(s)$ for a specific $s \in \R^d$. Combining Algorithm~\ref{algo_QLFT_2d_adaptive} with Grover adaptive search~\cite{GAS03,GAS05} allows us find $f^*(s')$ in time $O(\sqrt{N}\,\polylog(N))$, where $s'$ is the best approximation to $s$ on the adaptively chosen dual space.\footnote{Alternatively, we could find the $d$ best approximations, which takes time $O(\sqrt{N}\,\polylog(N))$, and then do a linear approximation.}
\end{remark}

\section{Classical and quantum computational complexity} \label{sec_optimality}
In this section we analyze the computational complexity of several natural tasks related to the computation of the discrete LFT. The results can be summarized as follows: (i) we show that no classical algorithm to compute the discrete LFT can be efficient, and that sampling LFT values is also classically difficult; (ii) quantum algorithms can be at most quadratically faster compared to the best classical algorithm for the task of computing the discrete LFT at a single point, and the algorithm proposed in this paper achieves this quadratic speedup in some cases; (iii) the presented quantum algorithm features the optimal scaling with respect to the parameter $W$, which we show to be a scale-independent parameter capturing the difficulty of a problem instance. 

We distinguish two basic computational tasks involving the discrete LFT. The first task is the computation of $f^*(s)$ for a given $s \in \cS_K^d$ and a given primal discretization $\cX^d_N$. The second task is to output $f^*(s)$ for $s$ chosen uniformly at random from $\{0,1\}^d$ and a given primal discretization $\cX^d_N$. The motivation for the sampling task is to analyze a ``classical equivalent'' of the construction of the quantum state \smash{$\frac{1}{\sqrt{K}} \sum_{j=0}^{K-1}\ket{j}\ket{f^*(s_j)}$}, followed by a measurement in the computational basis; such a quantum state is exactly the one constructed by Algorithm~\ref{algo_QLFT_2d}. We fix the dual space $\cS^d_K = \{0,1\}^d$ for simplicity (implying $K=2^d)$, but one could equivalently consider any scalar multiple of this set. We prove hardness of these tasks below, showing that their efficient solution would contradict known lower bounds on the unstructured search problem.

\begin{proposition}[Classical hardness of discrete LFT]
  \label{prop_dltlowerbound}
  Let $\cX_N^d = \{0,1\}^d$ with $N=2^d$. Any classical
  algorithm that outputs the value of the discrete LFT
  of a function $f$ satisfying Assumption~\ref{ass_regularity} over $\cX_N^d$ at an arbitrary dual value
  $s \in \cX_N^d$ requires $\Omega(2^d/d)$ queries to $f$ in general.
\end{proposition}
\begin{proof}
Let $z \in \{0,1\}^d$, and let $g : \{0,1\}^d \to \{0,1\}$ be such that $g(z) = 1$ and $g(y) = 0$ for all $y \in \{0,1\}^d \setminus \{z\}$. It is well known that $\Omega(2^d)$ queries to $g$ are necessary to determine $z$.\footnote{This is the worst case for a deterministic algorithm, and an expectation for any randomized algorithm that finds the answer with constant probability.} 
Let $f:[0,1]^d \to \R$ be defined by
\begin{align*}
f(x) := \max_{i=1,\dots,d} |x_i - z_i| \, ,    
\end{align*}  
which is a jointly convex function. To see this we recall that a function is jointly convex if the epigraph is a convex set~\cite{rockafellar70}. All the functions $|x_i - z_i|$ have a convex epigraph and hence taking the intersection remains a convex set.
Notice that on the vertices of the hypercube $[0,1]^d$, the equality $f(x) = 1-g(x)$ holds, and therefore we can simulate a call to $f(x)$ with a single query to $g(x)$.

Fix $\cX_N^d = \cS_N^d = \{0,1\}^d$ and let $e_j \in \{0,1\}^d$ be the string with a $1$ in position $j$ and $0$ otherwise. Consider the value of the discrete LFT at $e_j$, i.e.,
  \begin{align*}
f^*(e_j) 
= \max_{x \in \cX_N^d} \{\langle e_j, x \rangle - f(x)\} 
= \max_{x \in \cX_N^d} \{x_j - 1 + g(x)\} 
= \begin{cases} 1 & \text{if } z_j = 1 \\ 0 & \text{if } z_j = 0 \, , \end{cases} 
  \end{align*}
which implies that $f^*(e_j)=z_j$.  
It follows that if we are able to output the value of $f^*$ at arbitrary values of the dual variables, we simply need to evaluate $f^*(e_j)$ for $j=1,\dots,d$ to be able to determine $z$. As stated earlier, determining $z$ requires $\Omega(2^d)$ queries to $f$. Thus, to output the value of $f^*$ at a single point, must require $\Omega(2^d/d)$ queries to $f$.
\end{proof}

\begin{proposition}[Classical hardness of sampling discrete LFT]\label{prop_sampling_hard}
Let set $\cX_N^d=\{x_0,\ldots,x_{N-1}\} = \{0,1\}^d$ satisfy Assumption~\ref{ass_discrete_d_dim}, $\cS_K^d=\{s_0,\ldots,s_{K-1}\} = \{0,1\}^d$ such that $N=K=2^d$, and let $f:[0,1]^d \to \R$ be convex. Then any classical algorithm to sample a pair $(s_j, f^*(s_j))$, where $s_j$ is chosen uniformly at random from $\cS_K^d$, requires $\Omega(2^d/d)$ queries to $f$ in general. 
\end{proposition}
\begin{proof}
We define a function $f:[0,1]^d \to \R$ as
\begin{align*}
f(x) := 2^d \max_{i\in \{1,\dots,d\}} |x_i - z_i|\, ,    
\end{align*}  
for $z \in \{0,1\}^d$, similarly to Proposition~\ref{prop_dltlowerbound}. Notice that
\begin{equation*}
    f^*(s_j) = \max_{x \in \{0,1\}^d} \{\langle x, s_j \rangle - f(x) \} = \langle z, s_j \rangle - f(z) = \langle z, s_j \rangle\, ,
\end{equation*}
because $\langle x, s_j \rangle < d$ for all $x \in \{0,1\}^d$, and therefore the optimum of the $\max$ is always obtained by setting $x = z$, yielding a value of at least $0$ (for $x \neq z$, the expression has a negative value). Suppose we could sample $(s_j, f^*(s_j))$ where $s_j$ is chosen uniformly at random from $\cS_K^d$, by performing $q$ queries to $f$. If we gather $h$ such samples, thereby obtaining a set of pairs $\{s_j, \langle z, s_j \rangle \}_{j=1}^{h}$, we can set up a system of linear equations:
\begin{equation*}
    \langle s_j, x \rangle = \langle z, s_j \rangle\, , \qquad \text{for} \quad j=1,\dots,h\, .
\end{equation*}
This is a set of equations in the unknown $x \in \{0,1\}^d$. Whenever this system admits a unique solution, that solution has to be $z$. Recall that each equation $\langle s_j, x \rangle = \langle z, s_j \rangle$ is obtained by uniformly sampling $s_j \in \cS_K^d$. We need to determine what value of $h$ (i.e., how many equations) is needed, so that the system admits a unique solution. It is known that after sampling $d + t$ such equations, the probability that the system contains at least $d$ linearly independent equations is at least $1 - \frac{1}{2^t}$ \cite[Appendix~G]{mermin07quantum}. Hence, for a given maximum probability of failure $\delta$, the algorithm that we just described identifies the string $z$ requiring $d + \log \frac{1}{\delta}$ samples. Recall from Proposition~\ref{prop_dltlowerbound} that any classical algorithm that identifies $z$ with fixed (constant) probability requires $\Omega(2^d)$ queries to $f$; our algorithm requires $O(d)$ samples to solve this problem. Thus, we must have $2^d \leq d q$, from which the claimed result follows.
\end{proof}

Propositions \ref{prop_dltlowerbound} and \ref{prop_sampling_hard} show that computing or sampling the discrete LFT is classically difficult. Because the proofs are based on reductions from the unstructured search problem, i.e., identifying an unknown $d$-digit binary string only by means of an oracle that can recognize the string, they imply quantum lower bounds as well. They are discussed next.

\begin{corollary}[QLFT at a single point]
  \label{cor_qdltlowerbound}
  Let $\cX_N^d = \{0,1\}^d$ with $N=2^d$. Any quantum
  algorithm that outputs the value of the discrete LFT
  of a function $f$ satisfying Assumption~\ref{ass_regularity} over $\cX_N^d$ at an arbitrary dual value
  $s \in \cX_N^d$ requires $\Omega(\sqrt{2^d}/d)$ queries to $f$ in general.
\end{corollary}
\begin{proof}
Follows immediately from Proposition~\ref{prop_dltlowerbound}, since the lower bound to solve the same unstructured search problem using a quantum algorithm is $\Omega(\sqrt{2^d})$ quantum queries to $g$ \cite{zalka1999grover}. 
\end{proof}

\begin{corollary}[QLFT calculation can be hard]\label{cor_qlft_hard}
Let sets $\cX_N^d=\{x_0,\ldots,x_{N-1}\} \subseteq [0,1]^d$ and $\cS_K^d=\{s_0,\ldots,s_{K-1}\} \subseteq \R^d$ satisfy Assumption~\ref{ass_discrete_d_dim} and let $f:[0,1]^d \to \R$ be convex. Then any quantum algorithm to compute
\begin{align} \label{eq_state_v}
\ket{v} = \frac{1}{\sqrt{N}} \sum_{j=0}^{N-1} \ket{j}\ket{f^*(s_j)}   
\end{align}
with constant probability requires $\Omega(\sqrt{2^d}/d)$ queries to $f$ in general. 
\end{corollary}
\begin{proof}
Follows immediately from Proposition~\ref{prop_sampling_hard}, because if we could create quantum state \eqref{eq_state_v}, a measurement in the computational basis would yield a random pair $(j, f^*(s_j))$, and we can easily recover $s_j$ from $j$. The square root in the lower bound comes from the quantum lower bound on unstructured search \cite{zalka1999grover}.
\end{proof}

Notice that Corollary~\ref{cor_qdltlowerbound} by itself does not rule out an efficient quantum algorithm to create the superposition \smash{$\frac{1}{\sqrt{K}} \sum_{j=0}^{K-1} \ket{j}\ket{f^*(s_j)}$}, because creating this quantum state does not imply that we are able to efficiently output the value of $f^*(s_j)$ for arbitrary $s_j$. However, Corollary~\ref{cor_qlft_hard} implies that, in general, there cannot exist a quantum algorithm for the LFT that scales polynomially in $d$. The QLFT algorithm proposed in this paper has running time $O((NW/K)^{d/2}\polylog(N,K))$, but crucially, $NW/K$ can be equal to 1, where $W$ is as defined in \eqref{eq_W}; the ``hard'' instance constructed in Corollary~\ref{cor_qlft_hard}, showing that no algorithm can construct \smash{$\frac{1}{\sqrt{K}} \sum_{j=0}^{K-1} \ket{j}\ket{f^*(s_j)}$} efficiently in general, has $NW/K=2^d$.

Our next result shows that the parameter $W$, used to characterize the probability of success of the regular QLFT algorithm, is a scale-independent parameter: for any given problem instance there is an infinite family of instances that share the same $W$ and have LFT values that can be mapped one-to-one. We use this result to show that the QLFT algorithm features the optimal scaling with respect to $W$, and that this allows us, when $W=1$, to output $f^*(s_j)$ at arbitrary $s_j$ with a running time that matches Corollary~\ref{cor_qdltlowerbound}, i.e., quadratically faster than the classical lower bound of Proposition~\ref{prop_dltlowerbound}. For these problem instances, the QLFT algorithm therefore achieves a quadratic speedup for the task of outputting the value of the LFT with respect to any classical algorithm, and it achieves an exponential speedup for the task of constructing the superposition \smash{$\frac{1}{\sqrt{K}} \sum_{j=0}^{K-1} \ket{j}\ket{f^*(s_j)}$} with respect to the trivial quantum implementation of any classical algorithm. On the other hand, when $W$ is large the QLFT cannot be efficient.
\begin{lemma}
\label{lem_W_scale_invariant}
Let sets $\cX_N^d=\{x_0,\ldots,x_{N-1}\} \subseteq \R^d$ and $\cS_K^d=\{s_0,\ldots,s_{K-1}\} \subseteq \R^d$ satisfy Assumption~\ref{ass_discrete_d_dim} with grid discretization parameters $\gamma_{\mathrm{x}}, \gamma_{\mathrm{s}}$. Let $f$ be convex. Let $\xi:=\max_{i \in\{1,\ldots,N-2\}}\{c_i-c_{i-1}\}/\gamma_{\mathrm{x}}$, yielding $W =\lfloor \xi \gamma_{\mathrm{x}}/\gamma_{\mathrm{s}}\rfloor$. Then we can construct a function $\tilde{f}$ and  discretizations $\tilde{\cX}_N^d =\{\tilde{x}_0,\ldots,\tilde{x}_{N-1}\}, \tilde{\cS}_K^d=\{\tilde{s}_0,\ldots,\tilde{s}_{K-1}\}$ with $\tilde{\gamma}_x = 1, \tilde{\xi} = 1, \tilde{\gamma}_s = \gamma_{\mathrm{s}}/(\xi \gamma_{\mathrm{x}})$ such that $f^*(s_j) = \xi \tilde{f}^*(\tilde{s}_j)$ for all $j\in [K]$. Furthermore, the value of $\tilde{W}$ for $\tilde{f}$ is equal to $W$.
\end{lemma}
\begin{proof}
Define $\tilde{f}(x):= f(x)/\xi$, and choose $\tilde{\cX}_N^d = \{x_i/\gamma_{\mathrm{x}} : i\in [N]\}$, $\tilde{\cS}_K^d = \{s_j/(\xi \gamma_{\mathrm{x}}) : j \in [K]\}$. By construction, $\tilde{\xi} = \tilde{\gamma}_x =1$ and $\tilde{\gamma}_s = \gamma_{\mathrm{s}}/(\xi \gamma_{\mathrm{x}})$. Using the scaling properties of the LFT~\cite[Section~11.A]{rockafellar2009variational} gives $\tilde{f}^*(\tilde{s}_j)= \tilde{f}^*(s_j/(\xi \gamma_{\mathrm{x}})) = f^*(\xi\gamma_{\mathrm{x}} s_j / (\xi \gamma_{\mathrm{x}}))/\xi$. For the final part of the statement, we simply note that the parameter $W$ remains unchanged, i.e.,
\begin{align*}
    \tilde{W} 
    = \left \lfloor \tilde \xi  \frac{\tilde \gamma_{\mathrm{x}}}{\tilde \gamma_{\mathrm{s}}} \right \rfloor
    =\left \lfloor \frac{1}{\tilde \gamma_{\mathrm{s}}} \right \rfloor
    =\left \lfloor \xi \frac{\gamma_{\mathrm{x}}}{\gamma_{\mathrm{s}}} \right \rfloor
    =W.
\end{align*}
\end{proof}
The main consequence of Lemma~\ref{lem_W_scale_invariant} is the fact that we can always assume $\gamma_{\mathrm{x}} = 1$ and $\xi = 1$ when analyzing a problem instance: the instance can always be rescaled as a preprocessing step if necessary. This plays a role in the next result. To improve its readability let $\underline{x}:=\min_{i \in [N]} x_i$, $\overline{x}:=\max_{i \in [N]} x_i$, $\underline{s}:=\min_{j \in [K]} s_j$, $\overline{s}:=\max_{j \in [K]} s_i$, $\underline{c}:=\min_{i \in [N]} c_i$, and $\overline{c}:=\min_{i \in [N]} c_i$. 

\begin{corollary}[Fundamental importance of parameter $W$]\label{cor_W_dependence}
Let sets $\cX_N^d=\{x_0,\ldots,x_{N-1}\} \subseteq \R^d$ and $\cS_K^d=\{s_0,\ldots,s_{K-1}\} \subseteq \R^d$ satisfy Assumption~\ref{ass_discrete_d_dim}. Let $f$ be convex with discrete gradients $c_0,\dots,c_N$ defined in~\eqref{eq_ci}. Suppose that a quantum algorithm computes the state $\ket{v}$ defined in~\eqref{eq_state_v} with constant probability of success in time $\cC(N, K, d, \underline{x}, \overline{x}, \underline{s}, \overline{s}, \underline{c}, \overline{c})$ for some function $\cC$. Then if $\cC(N, K, d, \underline{x}, \overline{x}, \underline{s}, \overline{s}, \underline{c}, \overline{c}) = \poly(\log N, \log K, d)  \hat{\cC}(\underline{x}, \overline{x}, \underline{s}, \overline{s}, \underline{c}, \overline{c})$ for some function $\hat{\cC}$, we have $\hat{\cC}(\underline{x}, \overline{x}, \underline{s}, \overline{s},$ $\underline{c}, \overline{c}) = \Omega(\sqrt{W}/\poly(\log N, \log K, d))$ in general.
\end{corollary}
\begin{proof}
From the proof of Corollary~\ref{cor_qlft_hard} (using the same notation), such an algorithm requires $\Omega(\sqrt{2^d} / d)$ queries. By assumption, the running time is $\poly(\log N, \log K, d) \, \hat{\cC}(\underline{x}, \overline{x}, \underline{s}, \overline{s}, \underline{c}, \overline{c})$ for some function $\hat{\cC}$. Hence, we must have $\hat{\cC}(\underline{x}, \overline{x}, \underline{s}, \overline{s}, \underline{c}, \overline{c}) = \Omega(\sqrt{2^d}/\poly(\log N, \log K, d))$. Furthermore, notice that for the instance constructed in Corollary~\ref{cor_qlft_hard}, we can rewrite $\hat{\cC}(\underline{x}, \overline{x}, \underline{s}, \overline{s}, \underline{c}, \overline{c}) = \cC_{\gamma}(\gamma_{\mathrm{x}}, \gamma_{\mathrm{s}}, \xi) = \cC_W(W)$ for some functions $\cC_\gamma$ and $\cC_W$, where the first equality is due to the fact that $\gamma_{\mathrm{x}}, \gamma_{\mathrm{s}}, \xi$ subsume the arguments of $\hat{\cC}$, and the second equality is due to Lemma~\ref{lem_W_scale_invariant}. Recalling that $W = 2^d$ for the considered function, we have $\cC_W(2^d) = \Omega(\sqrt{2^d}/\poly(\log N, \log K, d))$, which concludes the proof.
\end{proof}
We remark that the $d$-dimensional regular QLFT algorithm, whose running time is characterized in Theorem~\ref{thm_QLFT} and  Corollary~\ref{cor_d_dim_QLFT_regular}, achieves the optimal scaling of $\sqrt{W}$ in terms of the parameter $W$ when combined with amplitude amplification~\cite{brassard02}.
As shown by Corollary~\ref{cor_W_dependence}, the parameter $W$ determines if an efficient QLFT is possible. We note that in case the function $f$ satisfies Assumption~\ref{ass_regularity_fac} then $W \leq \kappa K/N$, i.e., in case $K/N=O(1)$ the condition number $\kappa$ describes the optimal scaling of any QLFT algorithm. 
The following two examples show how $W$ may behave for different convex function $f:[0,1]^d \to \R$:
\begin{enumerate}[(i)]
    \item For a multivariate quadratic function $f:[0,1]^d \to \R$ given by $x \mapsto x^{\mathrm{T}}Qx+\langle a,x\rangle +b$ for some positive definite $Q\in \R^{d\times d}$, $a\in \R^d$, and $b \in \R$ we have $W \leq \kappa K/N=O(\kappa)=O(\lambda_{\max}(Q)/\lambda_{\min}(Q))$, under the assumption that $K/N=O(1)$ where $\lambda_{\max}(Q)$ and $\lambda_{\min}(Q)$ denote the maximal and minimal eigenvalue of $Q$. We used that for multivariate quadratic case the condition number of $f$ coincides with the condition number of the matrix $Q$. 
    \item For a piecewise linear convex function, we have $W\leq \eta \gamma_{\mathrm{x}}/\gamma_{\mathrm{s}}=O(\eta)$, under the assumption that $\gamma_{\mathrm{x}}/\gamma_{\mathrm{s}}=O(1)$, where $\eta$ denotes the maximal difference of two consecutive slopes in one direction.
\end{enumerate}

\begin{remark}[Exponential speedup for expectation values]\label{rmk_expspeedup}
 Suppose we are given a function $f:[0,1]^d \to \R$ satisfying Assumption~\ref{ass_regularity}, with condition number $\kappa=1$, and such that the ratio $\min_j |f^*(s_j)|/\max_j |f^*(s_j)|$ is bounded by a constant (i.e., independent of $N$ and $K$). We can then use Algorithm~\eqref{algo_QLFT_2d} followed by a digital-amplitude conversion (Remark~\ref{rmk_QDA_conversion}) to compute 
 \begin{align*}
    \frac{1}{\sqrt{K}} \sum_{j=0}^{K-1} \ket{j} \ket{f^*(s_j)} \qquad \textnormal{and} \qquad    \frac{1}{\sqrt{\alpha}} \sum_{j=0}^{K-1} f^*(s_j) \ket{j} =:\ket{f^*}\, ,
\end{align*}
in time $O(\polylog(N,K))$, where $\alpha:=\sum_{j=0}^{K-1} f^*(s_j)^2 $ is a normalization constant. 
Let $H \in \C^{K \times K}$ be an observable that can be implemented in $O(\polylog(N, K))$ time, e.g., a sum of tensor products of Pauli matrices with a small number of terms. We can then compute $\bra{f^*} H \ket{f^*}$ in time $O(\polylog(N,K))$, which is exponentially faster than any classical algorithm because it requires access to all elements of the superposition.  
\end{remark}


\paragraph{Acknowledgements}
We thank Peyman Mohajerin Esfahani for discussions on the discrete Legendre-Fenchel transform.
\bibliographystyle{arxiv_no_month}
\bibliography{bibliofile}

\end{document}